\documentclass[10pt,journal,compsoc]{IEEEtran}

\usepackage{amsthm}
\usepackage{graphicx}
\usepackage{textcomp}
\usepackage{xcolor}
\usepackage{amssymb}
\usepackage{amsfonts}
\usepackage{fdsymbol}
\newtheorem{definition}{Definition}
 
\newtheorem{theorem}{Theorem}

\usepackage[numbers,sort&compress]{natbib}
\usepackage{bm}
\usepackage{subfigure} 
\makeatletter
\newif\if@restonecol
\makeatother

\usepackage[linesnumbered,ruled,vlined]{algorithm2e}
\usepackage{algpseudocode}
\usepackage{amsmath}
\usepackage[ruled,linesnumbered]{algorithm2e}
\usepackage{setspace}
\newcommand{\nop}[1]{}

\ifCLASSINFOpdf
\else
\fi
\begin{document}
%
\title{DP2-Pub: Differentially Private High-Dimensional Data Publication with Invariant Post Randomization}

\author{Honglu~Jiang,~\IEEEmembership{Member,~IEEE,}
        Haotian~Yu,
        Xiuzhen~Cheng,~\IEEEmembership{Fellow,~IEEE,}
        Jian~Pei,~\IEEEmembership{Fellow,~IEEE,}
        Robert Pless,~\IEEEmembership{Member,~IEEE,}
        and Jiguo~Yu,~\IEEEmembership{Fellow,~IEEE,}
\IEEEcompsocitemizethanks{\IEEEcompsocthanksitem H. Jiang is with the Department of Computer Science and Software Engineering, Miami University, Oxford, OH 45056, USA, and the Department of Computer Science, The George Washington University, Washington, DC 20052, USA. E-mail: jiangh34@miamioh.edu.
\IEEEcompsocthanksitem H. Yu is with the Department of Data Analytics, The George Washington University, Washington, DC 20052 USA. E-mail: yuxx6789@gwu.edu.

\IEEEcompsocthanksitem X. Cheng is with the Department of Computer Science, The George Washington University, Washington, DC 20052 USA, and the School of Computer Science and Technology, Shandong University, Qingdao 266510, China. E-mail: xzcheng@sdu.edu.cn.

\IEEEcompsocthanksitem J. Pei is with the Departments of Computer Science, Biostatistics and Bioinformatics, and Electrical and Computer Engineering, Duke University, Durham, NC 27708, USA. E-mail: j.pei@duke.edu.

\IEEEcompsocthanksitem R. Pless is with the Department of Computer Science, The George Washington University, Washington, DC 20052 USA. E-mail: pless@gwu.edu.
\IEEEcompsocthanksitem J. Yu (Corresponding Author) is with Big Data Institute, Qilu University of Technology, Jinan, 250353, P.R. China, and Shandong Fundamental Research Center for Computer Science, Qilu University of Technology, Jinan, Shandong, 250353, P.R. China. E-mail: jiguoyu@sina.com.}

\thanks{Manuscript received; revised.}}
%
%

\markboth{Journal of ,~Vol.~, No.~, ~2022}%
{Shell \MakeLowercase{\textit{et al.}}: Bare Demo of IEEEtran.cls for Computer Society Journals}
%



\IEEEtitleabstractindextext{%
\begin{abstract}
A large amount of high-dimensional and heterogeneous data appear in practical applications, which are often published to third parties for data analysis, recommendations, targeted advertising, and reliable predictions. However, publishing these data may disclose personal sensitive information, resulting in an increasing concern on privacy violations. Privacy-preserving data publishing has received considerable attention in recent years. Unfortunately, the differentially private publication of high dimensional data remains a challenging problem. In this paper, we propose a differentially private high-dimensional data publication mechanism (DP2-Pub) that runs in two phases: a Markov-blanket-based attribute clustering phase and an invariant post randomization (PRAM) phase. Specifically, splitting attributes into several low-dimensional clusters with high intra-cluster cohesion and low inter-cluster coupling helps obtain a reasonable allocation of privacy budget, while a double-perturbation mechanism satisfying local differential privacy facilitates an invariant PRAM to ensure no loss of statistical information and thus significantly preserves data utility. We also extend our DP2-Pub mechanism to the scenario with a semi-honest server which satisfies local differential privacy. We conduct extensive experiments on four real-world datasets and the experimental results demonstrate that our mechanism can significantly improve the data utility of the published data while satisfying differential privacy. 
\end{abstract}

}

\maketitle

\IEEEdisplaynontitleabstractindextext

%
\IEEEpeerreviewmaketitle

\IEEEraisesectionheading{\section{Introduction}\label{sec:introduction}}

%
%
%
%
\IEEEPARstart{T}{he} rapid development of information technology has opened up the era of big data. Collecting and publishing an unprecedented amount of data, as well as mining data correlations and generating insights have become an important component of social statistical research \cite{GNAH}. A large amount of high-dimensional and heterogeneous data appear in various applications, which are often published to third parties for data analysis, recommendations, targeted advertisements, and reliable predictions. Examples include healthcare data, social networking data, Internet of Things data (i.e., IoT device monitoring data, location data, trajectory data), financial market data (i.e., electronics commercial data, credit card data), which can be used to dig out valuable information hidden behind the massive data for modern life. However, publishing these data may disclose personal sensitive information, resulting in an increasing concern of privacy violations. Privacy-preserving data publishing (PPDP) has gained significant attentions in recent years  as a promising approach for information sharing while preserving data privacy \cite{FBW}.

Generally speaking, commonly used approaches for PPDP can be characterized into three categories: encryption technology \cite{GRM}, $k$-anonymity \cite{SLA} and its derivative approaches ($l$-diversity \cite{MKV}, $t$-closeness \cite{LLV}), and differential privacy \cite{DCD}. Differential privacy has gradually become the \emph{de facto} standard privacy definition and provides a strong privacy guarantee. It rests on a sound mathematical foundation with a formal definition and rigorous proof while making the assumption that an attacker has the maximum background knowledge.

However, the differentially private publication of high dimensional data remains a challenging problem -- it suffers from the ``Curse of High-Dimensionality'' \cite{RYY}, that is, when the dimensionality increases, the complexity and cost of multi-dimensional data processing and analysis increases exponentially. Specifically, this curse is manifested in two aspects: first, since the high-dimensional data space is usually sparse, high dimensions and large attribute domains lead to a low “Signal-to-Noise Ratio" and low data utility; second, complex correlations exist between high-dimensional data, therefore the change of a single record may have a great impact on query results, leading to increased sensitivity.

To address these challenges, an effective way is to decompose high-dimensional data into a set of low-dimensional marginal tables along with inferring the joint distributions of the data, thus generating a synthetic dataset. A representative solution is PrivBayes \cite{ZCG}, which constructs a Bayesian network to model the data correlations and conditional probability distributions, allowing one to approximate the distributions of the original data using a set of low-dimensional marginal distributions. However, such an approach suffers from poor data utility and high communication cost, since too much noise is added when there are too many attribute pairs resulting in unreliable conditional probabilities. Moreover, most approaches generally ignore the different roles a dimension may play for a specific query – one dimension may be more important than another for a particular query. Additionally, one dimension may release more information than another if the same amount of noise is added; thus evenly allocating the total privacy budget to each dimension degrades the performance.

In this paper, we provide a two-phase mechanism (DP2-Pub) consisting of a Markov-blanket-based learning process and an invariant post randomization (PRAM) process satisfying local differential privacy to overcome the above difficulties. Our contributions can be summarized as follows:

\begin{itemize}

\item To capture the dependencies between the attributes in the dataset, we resort to differentially private Bayesian network construction, employing the exponential mechanism to attribute pairs using the mutual information as the score function.

\item We propose the procedure of attribute clustering with a Markov blanket learning algorithm based on the constructed Bayesian network. Our most fundamental purpose is to split attributes into several low-dimensional clusters with high intra-cluster cohesion and low inter-cluster coupling, thus obtaining a reasonable allocation of privacy budget determined by the conditional independence among attributes and the importance of each cluster.

\item Invariant PRAM is an important perturbation technique for privacy protection, which transforms each record stochastically in a dataset using delicately pre-selected probabilities. It ensures no loss of statistical information, thus can significantly preserve data utility. Motivated by this, we provide a double-perturbation mechanism to achieve invariant PRAM for two-valued and multivalued attributes and apply it to each attribute cluster. Resorting to the randomized mapping based post-processing property for differential privacy, we prove that the proposed double-perturbation mechanism satisfies differential privacy. 

\item To tackle the data privacy preservation problem for the scenario where each individual contributes a single data record to a semi-honest server, we extend our DP2-Pub mechanism to handle the high-dimensional data publication in a local-differential-privacy manner, 
in which each user locally perturbs its data satisfying local differential privacy, then the server conducts all the operations including attribute clustering and post randomization over the privatized data.

\item We evaluate the performance of data utility on four real-world datasets from two aspects, the total variation distance between the original dataset and the perturbed dataset and the classification error rate of SVM classification on the perturbed dataset. Experimental results indicate that our approach can obtain higher data utility of the published data compared with the state-of-the-art.

\end{itemize}

The rest of this paper is organized as follows. We provide a literature review in Section \ref{sec:rel}. Section \ref{sec:model} formulates our problem and presents necessary background knowledge on Bayesian network, differential privacy, random response and post randomization. In Section \ref{sec:our}, we propose our DP2-Pub mechanism by detailing the constructions of differentially private Bayesian network, attribute clustering, and invariant PRAM. Comprehensive experimental studies on four real-world datasets are presented in Section \ref{sec:eva}. Section \ref{sec:con} concludes the paper with a future research discussion.

\section{Related Work}\label{sec:rel}

Various differentially private mechanisms for high-dimensional data publications have been proposed in recent years. In this section, we briefly review the most relevant works from two perspectives: under centralized setting or distributed setting, and discuss how our work differs from the existing ones.

\subsection{Private Mechanisms Under Centralized setting}

A powerful approach of dimensionality reduction is the Bayesian network model proposed in \cite{ZCPC}, in which Zhang \emph{et al.} developed a differentially private scheme PrivBayes for publishing high-dimensional data. PrivBayes first constructs a Bayesian network to approximate the distribution of the original dataset, then adds noise into each marginal of the Bayesian network to guarantee differential privacy, next constructs an approximate distribution of the original dataset, and finally samples the tuples from the approximate distribution to construct a synthetic dataset. 

Researchers also have developed sampling techniques to support differentially private high-dimensional data publications. In \cite{CXZ}, Chen \emph{et al.} provided a solution to protect the joint distribution of the dimensions in a high-dimensional dataset compared with PrivBayes. They first established a robust sampling-based approach to investigate the dependencies over all attributes for constructing a dependence graph, then applied a junction tree algorithm to provide an inference mechanism for deriving the joint data distribution.
In \cite{LXJ}, Li \emph{et al.} proposed a differentially private data synthetization technique called DPCopula using Copula functions to handle multi-dimensional data. 
In \cite{XRZ}, Xu \emph{et al.} developed a high-dimensional data publishing algorithm under differential privacy to optimize the utility by first projecting a $d$-dimensional vector of user's attributes into a lower $k$-dimensional space using a random projection, then adding Gaussian noise to each resultant vector to obtain a synthetic dataset.

\subsection{Private Mechanisms Under Distributed setting}

The approaches mentioned above mainly consider centralized scenarios. Some efforts have also been devoted to differentially private high-dimensional data publications under distributed setting. Based on PrivBayes, Cheng \emph{et al.} \cite{CTS} considered a multi-party setting from multiple data owners and proposed a differentially private sequential update of the Bayesian network (DP-SUBN) approach, allowing the parties to collaboratively identify the Bayesian network that best approximates the joint distribution of the integrated dataset. 
Wang \emph{et al.} \cite{WBL} introduced a framework with a simple and generic aggregation and decoding technique. This framework can analyze, generalize and optimize several local differential privacy protocols \cite{EPK,BAS,FPE} for frequency estimation. In \cite{RYY}, Ren \emph{et al.} proposed a solution LoPub to realize high-dimensional data publication with local differential privacy in crowdsourced data publication systems. LoPub can first learn from the distributed data records to build correlations and joint distributions of attributes, then synthesize an approximate dataset achieving a good compromise between local differential privacy and data utility. In \cite{JGW}, Ju \emph{et al.} also considered the high-dimensional data publication problem under local differential privacy in the crowdsourced-sensing system. They proposed an aggregation and publication mechanism which provides local privacy guarantees for crowd-sensing users, approximates the statistical characteristics of high-dimensional perception data and publishes synthetic data. Wang \emph{et al.} \cite{WXY} proposed two mechanisms for collecting and analyzing users' private data under local  differential privacy, which can collect multidimensional data with both numerical and categorical attributes. In \cite{DFC}, Domingo-Ferrer developed several random-response-based complementary approaches for multi-dimensional data preservation. In \cite{TTC}, Takagi \emph{et al.} presented a privacy-preserving phased generative model (P3GM) for high-dimensional data, which employs a two-phase learning process for training the model to increase the robustness to the differential privacy constraint.

In light of the above analysis, the following aspects distinguish our work from the existing approaches. First, since the sensitivity of distinct dimensions are different and evenly allocating the total privacy budget to each dimension cannot obtain good performance, we take into account the privacy budget allocation problem to realize attribute clustering with a reasonable allocation of privacy budget. Second, we design an invariant PRAM mechanism instead of generating noisy conditional distributions of the Bayesian network, then apply it to each attribute cluster, which can significantly improve the data utility while satisfying local differential privacy.

\section{Problem Formulation and Preliminaries}\label{sec:model}

\subsection{Problem Formulation}

In this paper, we consider the following problem: a data server collects data containing a vast amount of individual information and aims to release an approximate dataset to third parties for their uses such as data analysis and recommendations. Let $D$ be the dataset, $n$ be the total number of records, and $\mathcal{A}=\{A_{1},A_{2}, \cdots A_{d}\}$ be the set of $d$ unique attributes. Assume that all attribute values are categorical as one can always discretize all numerical data. The domain of an attribute $A_{i}$ is denoted by $\Omega_{i}$, whose size is $|\Omega_{i}|$.

\subsection{Bayesian Network}

A Bayesian network is a type of probabilistic graphical model that approximately describes the joint distribution over a set of variables by specifying their conditional independence \cite{DKF}. More specifically, a Bayesian network is a directed acyclic graph (DAG) whose nodes represent attribute variables and edges model the direct dependence among attributes. Formally speaking, a Bayesian network $\mathcal{N}$ over $\mathcal{A}$ (the set of attributes in $D$) is defined as a set of $d$ attribute-parent (AP) pairs, $(A_1,\Pi_1), \cdots, (A_d,\Pi_d)$, where each AP contains a unique attribute and all its parent nodes in $\mathcal{N}$. If the maximum size of any parent set in $\mathcal{N}$ is $k$, we define $\mathcal{N}$ to be a $k$-degree Bayesian network. Let $Pr[\mathcal{A}]$ denote the joint distribution over all attributes in $D$. A Bayesian network $\mathcal{N}$ defines a way to approximate $Pr[\mathcal{A}]$ with $d$ conditional distributions $Pr[A_1|\Pi_1], Pr[A_2|\Pi_2], \cdots, Pr[A_d|\Pi_d]$, that is, $Pr_{\mathcal{N}}[\mathcal{A}]=\prod_{i=1}^{d}Pr[A_i|\Pi_i]$.
If $\mathcal{N}$ accurately captures the dependencies between the attributes in $\mathcal{A}$, $Pr_{\mathcal{N}}[\mathcal{A}]$ can be a good approximation to $Pr[\mathcal{A}]$. Moreover, the computation of $Pr_{\mathcal{N}}[\mathcal{A}]$ can be efficient and simple if the degree of $\mathcal{N}$ is small. Figure \ref{fig:baye} illustrates the Bayesian network over a set of five attributes, namely, \emph{age, gender, exposure to toxins, smoking}, and \emph{cancer}. Table \ref{tab:baye} shows the AP pairs in the sample Bayesian network.

\begin{figure}[!htb]
 \centering
  \includegraphics[width=0.25\textwidth]{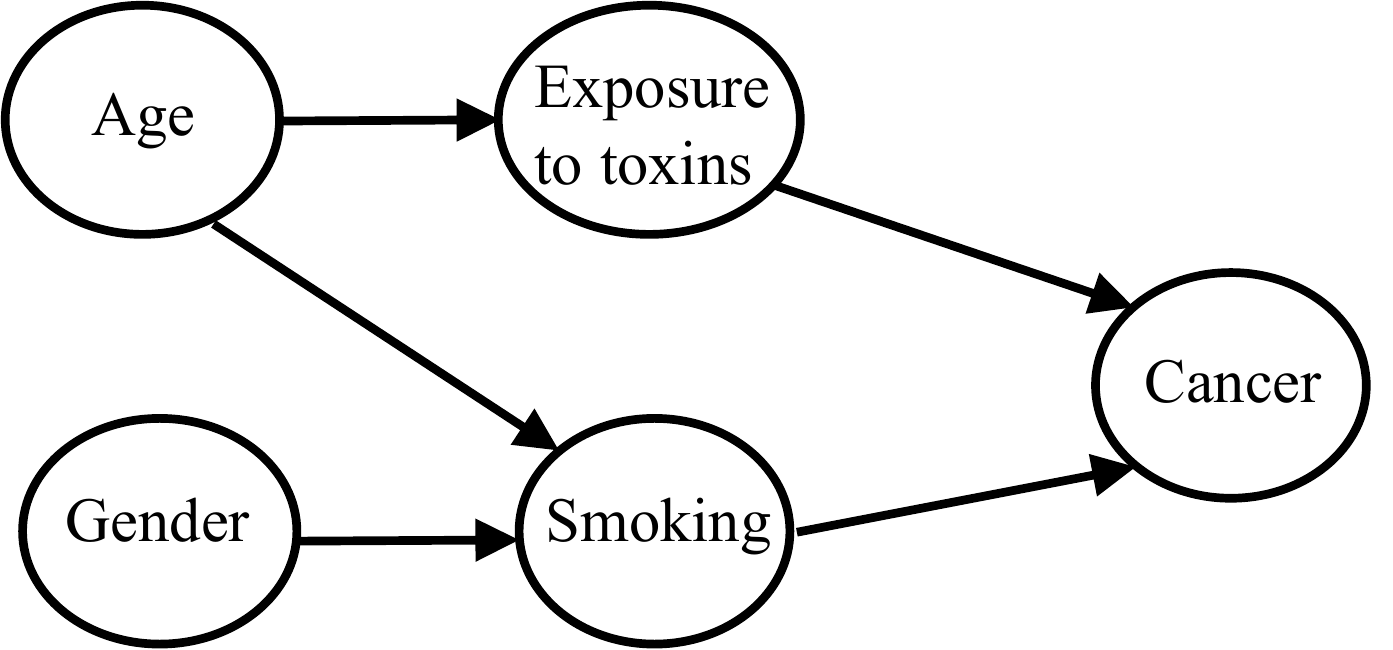}
  \caption{A Bayesian network over five attributes.}
  \label{fig:baye}
\end{figure}

\begin{table}[!htb]
\renewcommand{\arraystretch}{1.0}
\caption{The attribute-parent pairs in the Bayesian network shown in Fig.~\ref{fig:baye}} 
\centering
\scalebox{0.8}{
\begin{tabular}{cc}
\hline
\textbf{Attribute} \bm{$A_i$}                           &   \bm{$\Pi_i$}\\    
\hline 
\hline                       
Age                                                          & $\emptyset$          \\
\hline

Gender                                                       & $\emptyset$           \\
\hline

Exposure to toxins    & \{Age\}                         \\
\hline
Smoking                                                      & \{Age, Gender\}                 \\
\hline
Cancer          & \{Exposure to toxins, Smoking\}   \\
\hline
\end{tabular}}
\label{tab:baye}
\end{table}

 Given a dataset $D$, our goal is to construct a $k$-degree Bayesian network $\mathcal{N}$ which provides an accurate approximation to the full distribution of $D$. That is, $Pr_{\mathcal{N}}[\mathcal{A}]$ should be close to $Pr[\mathcal{A}]$. As $KL$-divergence \cite{TSAB} is commonly used as a measure of the similarity between a probability distribution and a candidate (estimated) distribution, in this paper, we adopt the $KL$-divergence of $Pr_{\mathcal{N}}[\mathcal{A}]$ and $Pr[\mathcal{A}]$ to measure the difference between these two probability distributions:
\begin{small}
\begin{equation}\nonumber
KL(Pr[\mathcal{A}],Pr_{\mathcal{N}}[\mathcal{A}])=\sum_{i=1}^{d}H(A_{i})-\sum_{i=1}^{d}I(A_{i},\Pi_{i})-H(\mathcal{A})
\end{equation}
\end{small}
where $H(A_i)$ denotes the entropy of the random variable $A_i$,
\begin{small}
\begin{equation}\nonumber
H(A_i)=-\sum\limits_{(x\in \Omega_i)}Pr[A_i=x] \log Pr[A_i=x]
\end{equation}
\end{small}
and $I(\cdot,\cdot)$ denotes the mutual information between the two variables:
\begin{small}
\begin{equation}\nonumber
\begin{aligned}
&I(A_i,\Pi_i)\\[0.9mm]
&=\sum\limits_{x\in \Omega_i} \! \sum\limits_{y \in dom(\Pi_i)}Pr[A_i =x,\Pi_i=y]\log\frac{Pr[A_i=x, \Pi_i=y]}{Pr[A_i=x]Pr[\Pi_i=y]}
\end{aligned}
\end{equation}
\end{small}
Here, $Pr[A_i,\Pi_i]$ is the joint distribution of $A_i$ and $\Pi_i$,  $Pr[A_i]$ and $Pr[\Pi_i]$ are the marginal distributions of $A_i$ and $\Pi_i$, respectively, and $H(\mathcal{A})$ is the joint entropy of all attribute variables in $\mathcal{A}$, which is defined as:
\begin{equation}\nonumber
\begin{aligned}
&H(\mathcal{A})=H(A_1,A_2,\cdots,A_d)\\
&\resizebox{1\hsize}{!}{$=-\sum\limits_{(x_1\in \Omega_1)}\cdots\sum\limits_{(x_d\in \Omega_d)} \!Pr[A_1\!=x_1,\cdots, A_d\!=x_d] \log Pr[A_1\!=x_1,\cdots,A_d\!=x_d]$}
\end{aligned}
\end{equation}
Therefore, learning a Bayesian network is to find $\mathcal{N}$ from $D$ with the minimum $KL(Pr[\mathcal{A}],Pr_{\mathcal{N}}[\mathcal{A}])$. The construction of $\mathcal{N}$ can be modeled as choosing a parent set $\Pi_{i}$ for each attribute $A_{i}$ to maximize $\sum_{i=1}^{d}I(A_{i},\Pi_{i})$ since $\sum_{i=1}^{d}H(A_{i})-H(\mathcal{A})$ is fixed once the dataset $D$ is given.

\subsection{Differential Privacy}

\emph{Differential privacy} (DP) has become the \emph{de facto} standard of privacy preservation, which ensures that query results of a dataset are insensitive to the change of a single record. 
Differential privacy is defined based on the neighboring datasets $D$ and $D'$, where $D'$ differs from $D$ by only one record:

\begin{definition}[Differential privacy \cite{DMN}] \label{def:DP}
A randomized algorithm $M$ is $\epsilon$-differentially private if for any pair of neighboring datasets $D$ and $D'$, and for all sets $S$ of possible outputs, we have
\begin{small}
\begin{equation}
Pr[M(D)\in S]\leq e^{\epsilon}Pr[M(D')\in S],\nonumber
\end{equation}
\end{small}
where $\epsilon$ is often a small positive real number. 
\end{definition}

 The smaller the $\epsilon$, the higher the level of privacy preservation. A smaller $\epsilon$ provides greater privacy preservation at the cost of lower data accuracy with more added noise. Differential privacy can be achieved by two best known mechanisms, namely the \emph{Laplace mechanism} \cite{DMN} and the \emph{exponential mechanism} \cite{MTK}. 
We provide the formal definition of the exponential mechanism as follows:

\begin{definition}[Exponential Mechanism \cite{MTK}] \label{def:exp}
 Given a random algorithm $M$ with the input dataset $D$ and the output entity object $o\in R$, where $R$ is the output range. Let $q(D,o)$ be the utility function and $\Delta q$ be the global sensitivity of function $q(D,o)$. If algorithm $M$ selects and outputs $o$ from $R$ at a probability proportional to $\exp(\frac{\epsilon q(D,o)}{2\Delta q})$, then $M$ is $\epsilon$-differentially private.
\end{definition}

DP techniques implicitly assume a trusted third party to collect data and thus can hardly be applied to the case where a server is not reliable. Therefore \emph{local differential privacy} (LDP) emerges, in which each user independently and locally conducts data perturbation. 
The formal definition of LDP can be shown as follows:

\begin{definition}[Local Differential Privacy \cite{DJW}] \label{ldp}
Consider $n$ records. A privacy algorithm $M$ with domain $Dom (M)$ and range $Ran (M)$ satisfies the $\epsilon$-local differential privacy if $M$ obtains the same output result $t^{*}$ $(t^{*}\subseteq Ran(M))$ on any two records $t$ and $t'$ $(t, t'\in Dom(M))$ with
$$Pr[M(t)=t^{*}]\leq e^{\epsilon}\times Pr[M(t')=t^{*}]$$
\end{definition}

Local differential privacy ensures the similarity between the output results of any two records. Random response (RR) \cite{WSL} is currently the most widely used technique for achieving local differential privacy.

\subsection{Random Response and Post Randomization}

Random response (RR) is a technique developed in social science to collect statistical data about individuals' sensitive information. Its main idea is to provide data privacy protection by making use of the uncertainty of responses to sensitive questions. 
Privacy comes from the randomness of the answers while accuracy comes from the noise generation procedure \cite{DAR}. 

Post randomization (PRAM) is another important perturbation technique for privacy protection, which stochastically transforms each record in a dataset using pre-selected probabilities. For a random variable $X$ with $s$ categories $c_1, c_2, \cdots,c_s$, let $\pi_i=Pr[X=c_i], i=1,\cdots,s$, and $\vec \pi=(\pi_1,\cdots,\pi_s)^\mathrm{T}$. The basic idea of PRAM is to select a transition probability matrix $P=(p_{ij})$ with $\sum_{j}p_{ij}=1$ for $i=1,2,\cdots,s$. Then the original category $c_i$ is changed to $c_j$ with probability $p_{ij}$. Let $Z$ denote the transformed variable. We have  $p_{ij}=Pr(Z=c_j|X=c_i)$, $\lambda_i=Pr[Z=c_i], i=1,\cdots,s$, and $\vec\lambda=(\lambda_1,\cdots,\lambda_s)^\mathrm{T}$.

Mathematically, PRAM is equivalent to RR.  Therefore many mathematical results developed for RR such as the local differential privacy guarantee can be applied to PRAM \cite{TSA}. In this paper, we employ PRAM for the case when a trusted data server is available (Section~\ref{sec:our}), where all the data can be processed at the server to maintain differential privacy,  and RR for the  case when the server is semi-honest, in which case local differential privacy is adopted for collecting data from each user to the server (Section~\ref{sec:extension}). 

\section{DP2-Pub With a Trusted Server}\label{sec:our}

In this section, we propose a novel \textbf{d}ifferentially \textbf{p}rivate high-dimensional data \textbf{pub}lication mechanism based on a \textbf{d}ouble-\textbf{p}erturbation process, namely DP2-Pub, assuming the availability of a trusted server that can access the original data. We first present an overview on DP2-Pub, then detail its modules in the following subsections. 

\subsection{Overview}\label{overview}

Figure \ref{fig:overview} illustrates the main procedure of DP2-Pub, which runs in two phases of attribute clustering and data randomization, with both being performed by the trusted server. Since both phases require access to the original dataset, we divide the total privacy budget $\epsilon$ into two portions with $\epsilon_1$ being used for the first phase and $\epsilon_2$ for the second phase, and demonstrate that the two phases are both differentially private. 

\begin{figure}[!htb]
 \centering
  \includegraphics[width=0.38\textwidth]{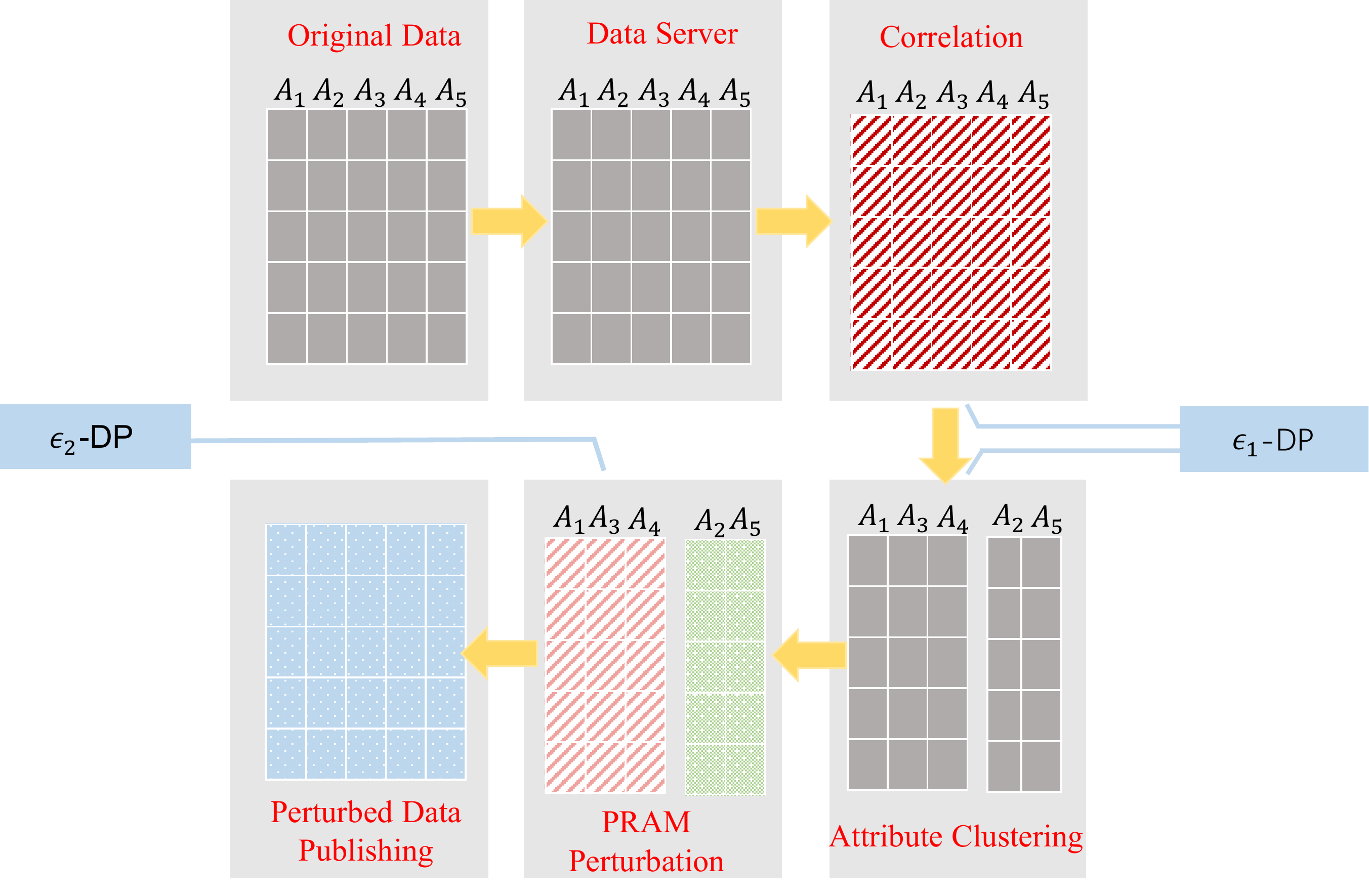}
  \caption{Overview of DP2-Pub.}
  \label{fig:overview}
\end{figure}

\textbf{1. Bayesian Network and Attribute Clustering.} 
To learn the correlations between different attribute variables, we adopt the approach of constructing a differentially private Bayesian network through the exponential mechanism presented in \cite{ZCG}. Based on the constructed Bayesian network, we propose the procedure of attribute clustering using the Markov blanket model to achieve high intra-cluster cohesion and low inter-cluster coupling. Each cluster is composed of a cluster head and its Markov blanket members, thus the attribute set can be divided into a number of disjoint clusters denoted as $CL_{1}, \cdots, CL_{t}$. Our most fundamental purpose is to realize attribute clustering, and then to obtain a reasonable allocation of privacy budget for each cluster based on its importance.  


\textbf{2. Data Randomization.} We propose an invariant post randomization method (PRAM) and apply it to each attribute cluster. Note that PRAM is an important technique for data perturbation, and that a PRAM  is invariant if the transition probability matrix $P$ satisfies $P\vec{\pi}=\vec{\pi}$ (except for the identity matrix $\mathbb{I}$). The appealing advantage of an invariant PRAM lies in that there is no loss of statistical information and thus can significantly preserve data utility. The key point to construct an invariant PRAM is to delicately solve $P$ satisfying $P\vec{\pi}=\vec{\pi}$. We design a double-perturbation mechanism to achieve the invariant property of PRAM for two-valued and multivalued attributes.

\subsection{Differentially Private Bayesian Network Construction}\label{bay}
 
In this section, we adopt the algorithm proposed in \cite{ZCG} to construct our Bayesian network in a differentially private manner, which employs the exponential mechanism to select $(A_{i},\Pi_{i})$, using the mutual information $I$ as the score function. \nop{ Note that the original algorithm in \cite{ZCG} selects the first attribute randomly, leading to a great uncertainty of the construction of Bayesian network, while in our algorithm we choose the one with the highest entropy. As the parent of the first selected attribute is always $\emptyset$, which attribute is selected does not impact on the differential privacy guarantee.} For completeness, we present the algorithm in Algorithm~\ref{ag1}.

\begin{algorithm}
\footnotesize
\caption{DP-Bayesian Network Construction}
\label{ag1}
\setstretch{1}
\LinesNumbered
\KwIn{Dataset $D$; $k$, the degree of the Bayesian network; and  the set of attributes $\mathcal{A}$}
\KwOut{Bayesian network $\mathcal{N}$ }
Initialize $\mathcal{N}=\emptyset$ and $V=\emptyset$;

Randomly select an attribute as $A_{1}$, add $A_{1}$ to $V$, add $(A_{1},\emptyset)$ to $\mathcal{N}$;

  \For{$i=2$ to $d$}
  {
     initialize $\Psi=\emptyset$;

 for ech $A_i\in \mathcal{A}\setminus V$ and each $\Pi_i \in \tbinom{V}{k}$, add $(A_i, \Pi_i)$ to $\Psi$; //$\tbinom{V}{k}$ denotes the set of all subsets of $V$ with size of $\min(k,|V|)$

Select an AP pair $(A_{i},\Pi_{i})$ from $\Psi$ at  a probability  proportional to $\exp(\frac{\frac{\epsilon_{1}}{d-1} \cdot I(A_{i},\Pi_{i})}{2\Delta})$; //$\Delta$ denotes the sensitivity of the mutual information function.

Add $A_{i}$ to $V$ and $(A_{i},\Pi_{i})$ to $\mathcal{N}$;
}

\textbf{return} $\mathcal{N}$
\end{algorithm}

At the beginning of Algorithm \ref{ag1}, we initialize the Bayesian network $\mathcal{N}$ to be an empty set. Let $V$ denote the set containing the attributes whose parent sets have been fixed and the initial set of $V$ is empty (Line 1). Then we randomly select an attribute from the $d$ attributes as the initial node and set its parent set empty (Line $2$). For each $A_i$, its AP pair is selected in a differentially private manner by the exponential mechanism (Lines $3$-$7$), of which the mutual information $I$ is taken as the score function, and $\Delta$ is its global sensitivity. Algorithm \ref{ag1} ensures that each invocation of the exponential mechanism satisfies $(\epsilon_1/(d-1))$-differential privacy, and the exponential mechanism is invoked $d-1$ times, so the construction of $\mathcal{N}$ is $\epsilon_1$-differentially private based on DP's composability property  (Lines $3-7$). 
We adopt the calculation of the sensitivity $\Delta$ of the mutual information in \cite{ZCG}, which is shown as follows:
\begin{equation}\label{eq1}
\Delta=\begin{cases}
    \frac{1}{n}\log(n)+\frac{n-1}{n}\log(\frac{n}{n-1}), &\text{if $A_i$ or $\Pi_i$ is binary,}\\[1.5mm]
   \frac{2}{n}\log(\frac{n+1}{2})+\frac{n-1}{n}\log(\frac{n+1}{n-1}), &\text{otherwise.} 
  \end{cases}
\end{equation}

\subsection{Attribute Clustering}\label{att}

Given a Bayesian network, the Markov blanket $MB(x)$ of an attribute variable $x$ can be intuitively represented as the set of parent nodes $Pa(x)$ and child nodes $Ch(x)$ of $x$ as well as the set of parent nodes of $x$'s child nodes, which can be formalized as follows:
\begin{small}
\begin{equation}\nonumber
MB(x)=Pa(x)\bigcup Ch(x)\bigcup \{Pa(y)|y\in Ch(x)\}
\end{equation}
\end{small}
We propose a procedure of attribute clustering shown in Algorithm \ref{ag2}. First, we initialize the set $S$ to include all the attributes of $\mathcal{N}$. Then we randomly select an attribute $x$, add $MB(x)$ and $x$ into a cluster, and delete all these attributes from $S$. Repeat this procedure until $S$ is empty. Each cluster is composed of a cluster head (an attribute variable $x$) and its Markov blanket members. Thus, the attribute set can be divided into a number of disjoint clusters, which can be denoted as $CL_{1}, \cdots, CL_{t}$. 

\begin{algorithm}
\footnotesize
\caption{Attribute Clustering}
\label{ag2}
\setstretch{1}
\LinesNumbered
\KwIn{Bayesian Network $\mathcal{N}$}
\KwOut{Cluster $CL_{1}, CL_{2} \cdots, CL_{i}$}
 
$S$=Set of all attributes of $\mathcal{N}$;

i=0;

\While{$S \neq \emptyset$} {

i=i+1;

Randomly select the attribute $x$ in $S$ and let $CL_{i}=MB(x)\bigcup\{x\}$;

$S=S-CL_{i}$;
 
}
\textbf{return} $CL_{1}, CL_{2} \cdots, CL_{i}$
\end{algorithm}

The key to overcome the curse of dimensionality is to decompose high-dimensional data into a set of low-dimensional data based on the conditional independences of the data. Bayesian network and Markov blanket are the most widely used graphical models for identifying a minimal set of attributes with strong correlations. Specifically, for any attribute variable $A_i$ in the Bayesian network, its Markov blanket is the set of attributes which are strongly corelated to $A_i$, while the attributes not in $A_i$’s Markov blanket are loosely correlated with $A_i$ or even conditionally independent of $A_i$. Therefore, our clustering algorithm yields attribute clusters with high intra-cluster correlation (cohesion) and low inter-cluster coupling which can improve the accuracy of the estimated joint distribution of the data. Note that the input of Algorithm \ref{ag2} is the differentially private Bayesian network constructed from Algorithm \ref{ag1}, which guarantees that the operation of attribute clustering does not break differential privacy.

According to the attribute clustering process, a reasonable allocation of the privacy budget for the next data randomization phase is determined by the conditional independence among the attributes in a cluster and the importance of the cluster based on the probability distributions over the dataset.  
Thus we define the importance factor (CIF) of each cluster $CL_{i}, 1\leq i\leq t$, in \eqref{cif}, which measures the importance of each cluster. The higher the CIF, the more important the cluster.

\begin{small}
\begin{equation}\label{cif}
CIF(CL_{i})={\frac{\sum\limits_{A_{j}\in CL_{i}}H(A_{j})}{\sum\limits_{k=1}^{d}H(A_{k})}}
\end{equation} 
\end{small}

Based on the CIF, one can allocate a privacy budget to each cluster, following the principle stating that the smaller the privacy budget, the higher the level of privacy preservation. Therefore, we define the privacy budget coefficient (PBC) for each cluster $CL_i$ as follows:

\begin{small}
\begin{equation}\label{PBC}
PBC(CL_{i})=\frac{\frac{1}{CIF(CL_{i})}}{\sum\limits_{j=1}^{t}\frac{1}{CIF(CL_{j})}}
\end{equation} 
\end{small}

Note that \eqref{PBC} conducts a normalization of PBC so that it falls into the $[0,1]$ interval. As the privacy budget of the data randomization is $\epsilon_{2}$, which will be carried out on each cluster at the server,  the privacy budget allocated for cluster $CL_i$ is $PBC(CL_{i})\cdot \epsilon_{2}$.

\subsection{Invariant PRAM}\label{inv}
 
The characteristic of an invariant PRAM lies in that the transition probability matrix $P$ satisfying $P\vec{\pi}=\vec{\pi}$. 
In this section, we propose an invariant PRAM scheme, which is suitable for 
categorical attributes. 
The main idea of our approach is to compute $P$ via 
double-perturbation, as shown in Figure \ref{fig:double}. For an attribute variable $X$, let $X_1$ denote the perturbed variable after the first perturbation, and $X_2$ denote the one after the second perturbation. We first construct a transition probability matrix $Q=(q_{ij})$ satisfying differential privacy and conduct the first perturbation on the attribute variable $X$ according $Q$. Then we compute the estimate of $\pi$ based on the perturbed data $X_1$, denoted as $\hat{\vec\pi}$, construct the transition probability matrix $\tilde{Q}=(\tilde{q}_{ij})$ for the second perturbation according to a specific rule to achieve $\tilde{Q}\cdot Q\cdot\vec{\pi} =\vec{\pi}$, and finally carry out the second perturbation on $X_1$ based on $\tilde{Q}$ to obtain $X_2$. More specifically, the rule of constructing $\tilde{Q}$ is to set $\tilde{q}_{ji}=Pr(X=c_i|X_1=c_j)$,  where $\tilde{q}_{ji}$ denotes the probability of $X_1$ being changed from category $c_j$ to $c_i$ in the second perturbation. In other words, $\tilde{Q}$ can be considered as the inverse of $Q$ while ensuring that $\tilde{Q}$ is also a transition probability matrix. Thus, the double-perturbation is an invariant PRAM with $P=\tilde{Q}\cdot Q$, and it satisfies differential privacy since the first perturbation satisfies differential privacy and the second perturbation is a randomized mapping of the first one.

\begin{figure}[!htb]
 \centering
  \includegraphics[width=0.35\textwidth]{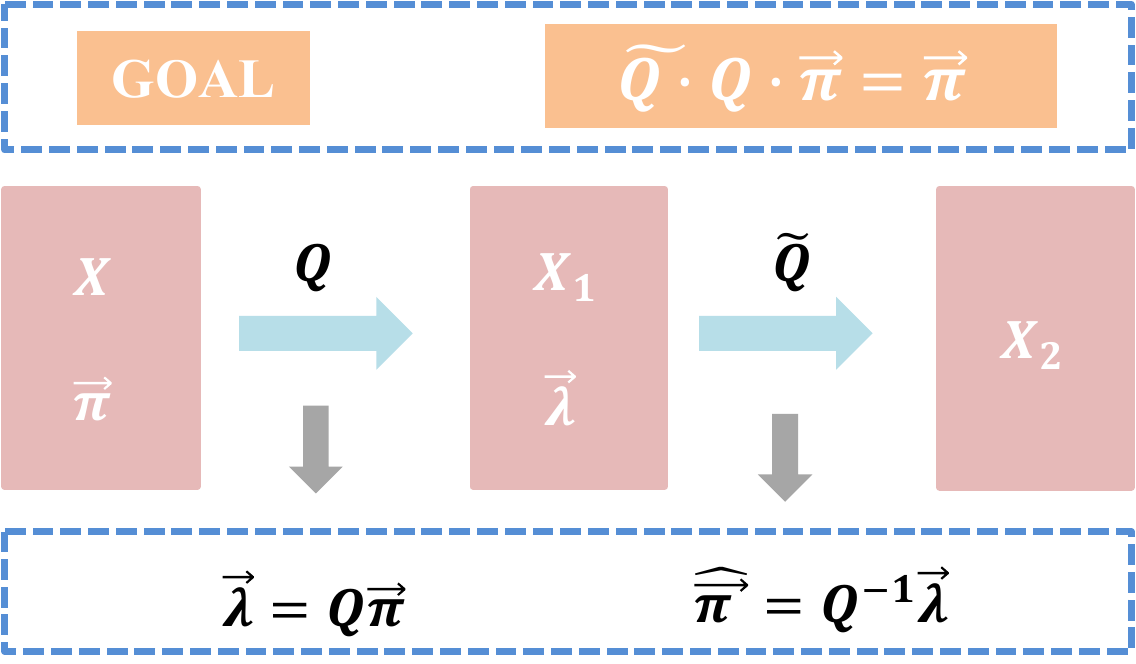}
  \caption{Process of Double Perturbation.}
   \label{fig:double}
\end{figure}

The attribute variables can be either two-valued or multivalued. For better elaboration, we detail the construction of the transition probability matrix for a two-valued attribute first and then extend the procedure to multivalued attributes. Following that, we apply the construction to compound variables and present the analysis on the privacy guarantee.

\subsubsection{Two-Valued Attributes}\label{sec:two-valued}

We start from the case of a categorical attribute variable $X$ with only two possible values of $c_1$ and $c_2$. Let $\pi_{1}=Pr[X=c_1]$ and $\pi_{2}=Pr[X=c_2]$; denoted by $\vec \pi=(\pi_1, \pi_2)^\mathrm{T}$. The data randomization process is conducted in a way that $c_1$ or $c_2$ either remains unchanged with probability $q$, or is changed to the other value with the probability of $1-q$; that is $q_{11}=q_{22}=q$ and $q_{12}=q_{21}=1-q$. Thus, the transition probability matrix $Q$ of the two-valued variable $X$ is a $2\times 2$ matrix, which can be shown as follows:

\begin{small}
\begin{equation}\nonumber
Q=\left[ \begin{array}{cc}
q &  1-q\\
1-q & q\\
\end{array} 
\right]
\end{equation}
\end{small}

Let $\lambda_{1}=Pr[X_1=c_1]$ and $\lambda_{2}=Pr[X_1=c_2]$; set $\vec\lambda=(\lambda_1, \lambda_2)^\mathrm{T}$. Note that the transition probability matrix satisfies $\vec\lambda=Q\vec\pi$.

In this setting, let $q=\frac{e^{\epsilon}}{1+e^{\epsilon}}$. Then the local differential privacy can be satisfied as:

\begin{small}
\begin{equation}\nonumber
\frac{Pr(X_1=c_1|X=c_1)}{Pr(X_1=c_1|X=c_2)}\leq\frac{q}{1-q}\leq e^{\epsilon}
\end{equation}
\end{small}


As mentioned earlier, we propose our invariant post randomization method to preserve the statistical information and data utility to the greatest possible extent, which first adopts transition probability matrix $Q$ on the original attribute variable $X$, estimates the probability distribution of the variable $X$ based on the perturbed data after the first perturbation, and constructs the transition probability matrix $\tilde{Q}$ for the second perturbation according to the transition probability matrix $Q$ and the probability distribution of the perturbed data $X_1$. The advantage of this double-perturbation mechanism lies in that there is no need to know the probability distribution of the original data in advance--we actually do not know the probability distribution of $X$--the transition probability matrix of the original data is thus constructed adaptively.
After obtaining the perturbed data $X_1$ with $\vec\lambda$, we can obtain the estimate of the original attribute variable distribution:

\begin{small}
\begin{equation}\label{est}
\hat{\vec\pi}=Q^{-1}\vec\lambda
\end{equation}
\end{small}
Then we compute the transition probability of each variable for the second perturbation as follows:

\begin{small}
\begin{equation}\nonumber
\tilde{q}_{11}=Pr(X=c_1|X_1=c_1)=\frac{\hat{\pi}_1\cdot q}{q\cdot\hat{\pi}_1+(1-q)\cdot\hat{\pi}_2}
\end{equation}
\end{small}
\begin{small}
\begin{equation}\nonumber
\tilde{q}_{22}=Pr(X=c_2|X_1=c_2)=\frac{\hat{\pi}_2\cdot q}{q\cdot\hat{\pi}_2+(1-q)\cdot\hat{\pi}_1}
\end{equation}
\end{small}
\begin{small}
\begin{equation}\nonumber
\begin{aligned}
\tilde{q}_{12}=Pr(X=c_2|X_1=c_1)=\frac{\hat{\pi}_2\cdot (1-q)}{q\cdot\hat{\pi}_1+(1-q)\cdot\hat{\pi}_2}
=1-\tilde{q}_{11}
\end{aligned}
\end{equation}
\end{small}
\begin{small}
\begin{equation}\nonumber
\begin{aligned}
\tilde{q}_{21}=Pr(X=c_1|X_1=c_2)=\frac{\hat{\pi}_1\cdot (1-q)}{q\cdot\hat{\pi}_2+(1-q)\cdot\hat{\pi}_1}
=1-\tilde{q}_{22}
\end{aligned}
\end{equation}
\end{small}

Accordingly, we obtain the transition probability matrix for the second perturbation:

\begin{small}
\begin{equation}\nonumber
\tilde{Q}=\left[ \begin{array}{cc}
\tilde{q}_{11} &  \tilde{q}_{12}\\[2mm]
\tilde{q}_{21} & \tilde{q}_{22}\\
\end{array} 
\right]
=\left[ \begin{array}{ccc}
\frac{\hat{\pi}_1\cdot q}{q\cdot\hat{\pi}_1+(1-q)\cdot\hat{\pi}_2} &  \frac{(1-q)\cdot\hat{\pi}_2}{q\cdot\hat{\pi}_1+(1-q)\cdot\hat{\pi}_2}\\[2mm]
\frac{(1-q)\cdot\hat{\pi}_1}{q\cdot\hat{\pi}_2+(1-q)\cdot\hat{\pi}_1} & \frac{\hat{\pi}_2\cdot q}{q\cdot\hat{\pi}_2+(1-q)\cdot\hat{\pi}_1}\
\end{array} 
\right]
\end{equation}
\end{small}

Therefore, to obtain the invariant PRAMed data $X_2$ of the attribute variable $X$, we apply $\tilde{Q}$ to the perturbed data $X_1$ during the second perturbation. These two phases of data perturbation with $Q$ and $\tilde{Q}$ successfully realize an invariant PRAM with $P=Q\cdot \tilde{Q}$, where $\tilde{Q}$ can be considered as the inverse of $Q$ while ensuring that $\tilde{Q}$ is also a transition probability matrix.

\subsubsection{Multivalued Attributes}\label{sec:multivalued}

The perturbation of the multivalued attributes is similar to that of the two-valued one. We consider a categorical random variable with $s$ possible values $c_1, c_2\cdots, c_s$. Let $\pi_{i}=Pr[X=c_i], i=1,\cdots, s$ and $\vec \pi=(\pi_1,\cdots, \pi_s)^\mathrm{T}$. Given that $X$ belongs to category $c_i$, it either remains unchanged with probability $q_{ii}$, or is changed uniformly at random with the probability of $\frac{1-q_{ii}}{s-1}$ to one of the other $s-1$ categories. 
That is, the transition probability matrix is a $s\times s$ one, which can be formalized as follows:

\begin{small}
\begin{equation}\nonumber
Q=\left[ \begin{array}{cccc}
q_{11}& \frac{1-q_{11}}{s-1}&  \cdots&\frac{1-q_{11}}{s-1} \\[1.2mm]
\frac{1-q_{22}}{s-1}& q_{22}&  \cdots&\frac{1-q_{22}}{s-1} \\[1.2mm]
\cdots& \cdots&\cdots&\cdots\\
\frac{1-q_{ss}}{s-1}& \frac{1-q_{ss}}{s-1} & \cdots&q_{ss}\\
\end{array} 
\right]
\end{equation}
\end{small}

To satisfy local differential privacy, we set $q_{11}=q_{22}=\cdots=q_{ss}=\frac{e^{\epsilon}}{s-1+e^{\epsilon}}$; then $Q$ can be denoted as:
\begin{small}
\begin{equation}\nonumber
Q=\left[ \begin{array}{cccc}
\frac{e^{\epsilon}}{s-1+e^{\epsilon}}& \frac{1}{s-1+e^{\epsilon}} & \cdots&\frac{1}{s-1+e^{\epsilon}} \\[1.2mm]
\frac{1}{s-1+e^{\epsilon}}& \frac{e^{\epsilon}}{s-1+e^{\epsilon}}&  \cdots&\frac{1}{s-1+e^{\epsilon}} \\[1.2mm]
\cdots& \cdots&\cdots&\cdots\\
\frac{1}{s-1+e^{\epsilon}}& \frac{1}{s-1+e^{\epsilon}} & \cdots&\frac{e^{\epsilon}}{s-1+e^{\epsilon}} \\
\end{array} 
\right]
\end{equation}
\end{small}
The PRAMed variable can be denoted as $_1$ after applying $Q$ to $X$. Correspondingly, let $\lambda_{i}=Pr[X_1=c_i], i=1,\cdots, s$, $\vec\lambda=(\lambda_1,\cdots, \lambda_s)^\mathrm{T}$. In this setting, the local differential privacy condition can be satisfied as:

\begin{small}
\begin{equation}\nonumber
\frac{Pr(X_1=c_i|X=c_i)}{Pr(X_1=c_i|X'\neq c_i)}\leq\frac{q_{ii}}{q_{ji}(j\neq i)}\leq e^{\epsilon}
\end{equation}
\end{small}
We can compute the estimated $\hat{\vec\pi}$ of the original attribute variable $X$ as 

\begin{small}
\begin{equation}\nonumber
\hat{\vec\pi}=Q^{-1}\cdot \vec\lambda
\end{equation}
\end{small}
Then the elements of the transition probability matrix $\tilde{Q}$ for the second perturbation can be computed as:

\begin{small}
\begin{equation}\nonumber
\begin{aligned}
\tilde{q}_{ij}=Pr(X=c_j|X_1=c_i)=\frac{\hat{\pi}_j\cdot q_{ji}}{\sum\limits_{k=1}^{s}\hat{\pi}_k\cdot q_{ki}}\\
\end{aligned}
\end{equation}
\end{small}
It can be observed that $\sum\limits_{j=1}^{s}\tilde{q}_{ij}=1$, which satisfies the property of a transition probability matrix. We take $\sum\limits_{j=1}^{s}\tilde{q}_{1j}$ as an example:
\begin{small}
\begin{equation}\nonumber
\begin{aligned}
&\sum\limits_{j=1}^{s}\tilde{q}_{1j}=\tilde{q}_{11}+\tilde{q}_{12}+ \cdots+\tilde{q}_{1s}\\[1.2mm]
&=\resizebox{0.96\hsize}{!}{$Pr(X=c_1|X_1=c_1)+Pr(X=c_2|X_1=c_1)+\cdots\!+\!Pr(X=c_s|X_1=c_1)$}\\[1.2mm]
&=\frac{\hat{\pi}_1\cdot q_{11}}{\sum\limits_{k=1}^{s}\hat{\pi}_k\cdot q_{k1}}+
\frac{\hat{\pi}_2\cdot q_{21}}{\sum\limits_{k=1}^{s}\hat{\pi}_k\cdot q_{k1}}+\cdots
\frac{\hat{\pi}_s\cdot q_{s1}}{\sum\limits_{k=1}^{s}\hat{\pi}_k\cdot q_{k1}}\\[1.2mm]
&=1
\end{aligned}
\end{equation}
\end{small}
After applying $\tilde{Q}$ to the perturbed data $X_1$ in the second perturbation, we obtain the invariant PRAM result $X_2$ of the attribute variable $X$. 

\subsubsection{Compound Variables}\label{compound}

Since an attribute cluster may include more than one two-valued or multivalued attribute variables which are strongly correlated, one can treat all these variables as a compound one. Thus an invariant PRAM for  
compound variables \cite{TSA} is needed, which first computes the transition probability matrix for each attribute variable, then computes the transition probability matrix for the compound one. For example, for two categorical variables $X$ with $r$ categories and $Y$ with $s$ categories, we may first compute the invariant PRAM transition probability matrix of $X$ and $Y$, denoted as $E=(e_{ij})$ and $F=(f_{ij})$, respectively. Then the combination of $X$ and $Y$ can 

\begin{small}
\begin{equation}\nonumber
\begin{aligned}
&Q_{EF}=E\otimes F=\\
&\left[\begin{array}{ccccccc}
e_{11}f_{11} & e_{11}f_{12}&\cdots e_{11}f_{1s}&&\ e_{1r}f_{11}& e_{1r}f_{12}&\cdots e_{1r}f_{1s}\\

 \vdots& \vdots& \vdots& \cdots &\vdots& \vdots& \vdots\\
e_{11}f_{s1} & e_{11}f_{s2}&\cdots e_{11}f_{ss}&&\ e_{1r}f_{s1}& e_{1r}f_{s2}&\cdots e_{1r}f_{ss}\\

 &\vdots &&\ddots&& \vdots\\
e_{r1}f_{11} & e_{r1}f_{12}&\cdots e_{r1}f_{1s}&&\ e_{rr}f_{11}& e_{rr}f_{12}&\cdots e_{rr}f_{1s}\\
 \vdots& \vdots& \vdots&\cdots& \vdots& \vdots& \vdots\\\vspace{0.2cm}
e_{r1}f_{s1} & e_{r1}f_{s2}&\cdots e_{r1}f_{ss}&&\ e_{rr}f_{s1}& e_{rr}f_{s2}&\cdots e_{rr}f_{ss}\\
\end{array} 
\right]
\end{aligned}
\end{equation}
\end{small}

Similarly, when there exist three categorical variables $X, Y, Z$ with  respectively $r,s,t$ categories, 
we may first compute the invariant PRAM transition probability matrix of $X$, $Y$ and $Z$, denoted as $E=(e_{ij})$, $F=(f_{ij})$ and $G=(g_{ij})$. Then the combination of $X$, $Y$ and $Z$ can be regarded as a compound variable with $r\cdot s \cdot t$ categories, whose transition probability $Q_{EFG}$ is the Kronecker product of $E$, $F$ and $G$, which is a $(r\cdot s \cdot t)\times (r\cdot s \cdot t)$ matrix.

As mentioned earlier, when applying the proposed invariant PRAM to each cluster, we need to allocate privacy budget $PBC(CL_{i})\cdot\epsilon_{2}$ to cluster $CL_i$. If a cluster $CL_i$ includes more than one attribute variable, we first compute the transition probability matrix for each attribute variable with uniformly allocated privacy budget $\frac{PBC(CL_{i})\cdot \epsilon_{2}}{|CL_i|}$, then compute the transition probability matrix for the compound variable.

\subsection{Privacy Analysis}\label{pri}

As discussed in Section \ref{bay}, Algorithm \ref{ag1} satisfies $\epsilon_1$-differential privacy, i.e., the procedure of Bayesian network construction satisfies differential privacy. The procedure of attribute clustering just simply cluster the attributes based on  the constructed Bayesian network, which does not disclose more information. Therefore, one can say that the first phase of DP2-Pub, i.e., Bayesian network construction and attribute clustering, satisfies $\epsilon_1$-differential privacy. 

Now we analyze the second phase, i.e., the phase of data randomization. According to \cite{DAR}, differential privacy is resistant to any randomized mapping of differentially private results. More specifically, with randomized mapping, a data analyst cannot 
make the output of a differentially private algorithm $M$ less differentially private without any additional knowledge about the private dataset \cite{DAR}. That is, if an algorithm is differentially private, simply conducting randomized mapping on the output of the algorithm without any additional knowledge does not leak any extra private information, which has been proved by the following Post-Processing Proposition \cite{DAR}:

\newtheorem{mypr1}{Proposition} 
\begin{mypr1}\label{pro1}
(Post-Processing \cite{DAR}) Let $M:D\rightarrow R$ be a randomized algorithm that is $\epsilon$-differentially private. Let $f: R\rightarrow R'$ be an arbitrary randomized mapping. Then $f\circ M: D\rightarrow R'$ is $\epsilon$-differentially private.
\end{mypr1}

\begin{theorem}
The double-perturbation of Invariant PRAM satisfies $\epsilon_2$-differential privacy.
\end{theorem}

\begin{proof}
The first perturbation of each attribute variable is a post randomization satisfying local differential privacy for both two-valued and multivalued variables:
\begin{small}
\begin{equation}\nonumber
\frac{Pr(X_1=c_1|X=c_1)}{Pr(X_1=c_1|X=c_2)}\leq\frac{q}{1-q}
\end{equation}
\end{small}
\begin{small}
\begin{equation}\nonumber
\frac{Pr(X_1=c_i|X=c_i)}{Pr(X_1=c_i|X'\neq c_i)}\leq\frac{q_{ii}}{q_{ji}(j\neq i)}
\end{equation}
\end{small}
of which $q$ and $q_{ii}$ are determined by the privacy budget allocated for each attribute variable in a cluster. According to Proposition \ref{pro1}, the second random perturbation of our invariant PRAM mechanism can be considered as a randomized mapping of the differentially private algorithm output, that is, a randomized mapping based post-processing of differential privacy. 

Since each cluster $CL_i$ is $PBC(CL_{i})\cdot \epsilon_{2}$-differentially private, the $t$ clusters can be regarded as a $t$-dimensional dataset achieving $\epsilon_2$-differential privacy according to the sequential composition theorem \cite{FDM}.
\end{proof}

Accordingly, one can obtain the following theorem.
\begin{theorem}
The DP2-Pub satisfies $(\epsilon_1+\epsilon_2)$-differential privacy according to sequential composition theorem \cite{FDM}.
\end{theorem}

\section{DP2-Pub With a Semi-honest Server}\label{sec:extension}


In many practical settings the central data server may not be trustworthy -- it is generally semi-honest, i.e., honest-but-curious, which faithfully follows the protocol but tries its best to infer as much knowledge as possible. Therefore in this section, we extend our DP2-Pub mechanism to consider a semi-honest server. A number of users generate multi-dimensional data records, then send them to a server who intends to release an approximate dataset to third-parties for various applications. Formally, each user contributes a data record constituting a dataset $D=\{U^1, U^2, \cdots U^n\}$, where $U^i$ denotes the data record of user $i$ and $n$ is the total number of records/users. 

Figure \ref{fig:distributed} illustrates the main procedure of DP2-Pub with a semi-honest server, which includes three main steps: privacy preservation of local data satisfying local differential privacy, Markov-blanket-based cluster learning based on Bayesian network, and the PRAM perturbation on the private data. Both the attribute clustering and PRAM perturbation are conducted at the data server, while the local differential privacy protection is performed by each user. Although the data server is semi-honest, it can only access the private data processed by each user. 

\begin{figure}[!htb]
 \centering
  \includegraphics[width=0.4\textwidth]{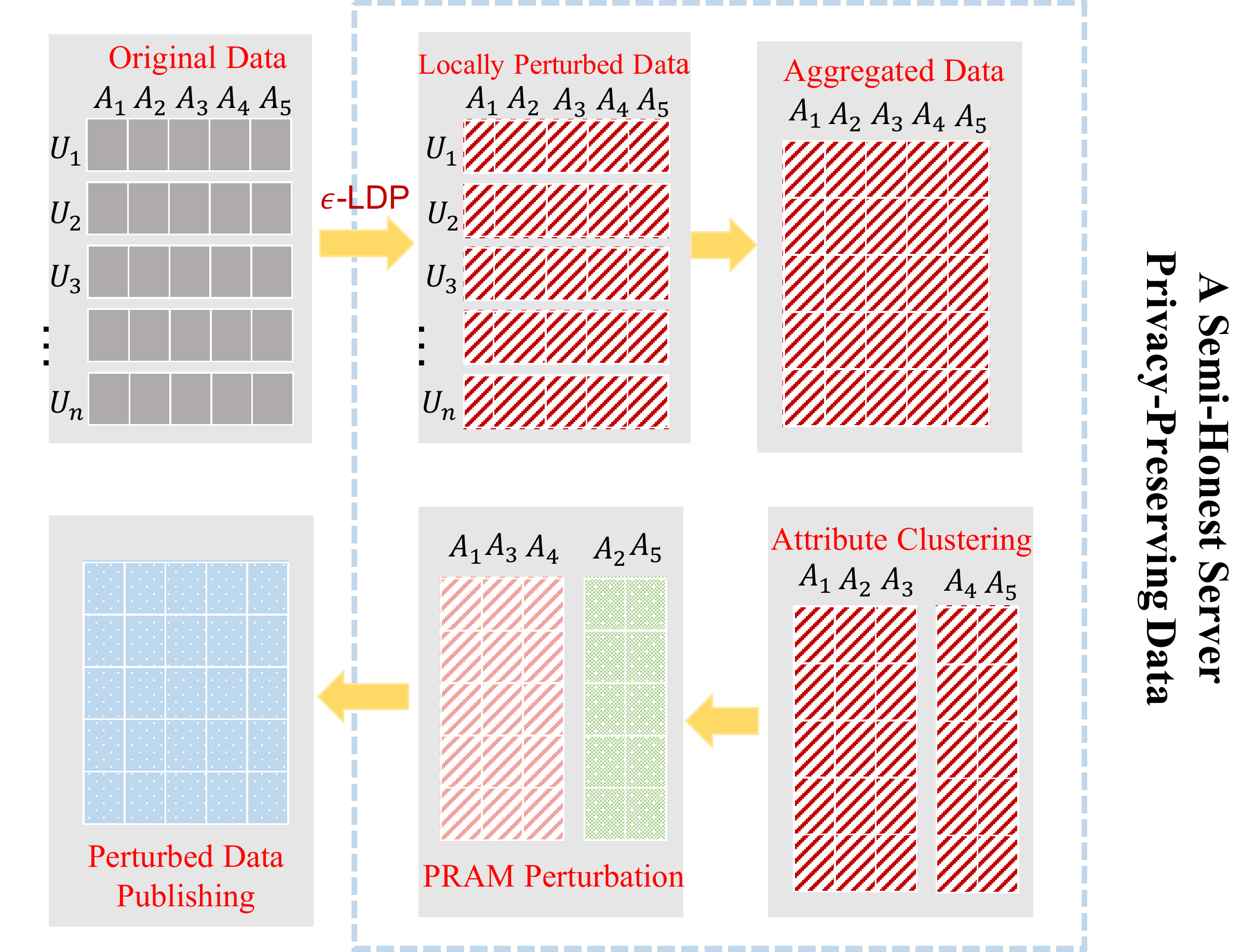}
  \caption{Overview of DP2-Pub with a Semi-honest Server.}
   \label{fig:distributed}
\end{figure}

We first propose a local randomization using RR on each user's data making it satisfy LDP, then the sanitized data is sent to and aggregated at the central server. Each user $i$ has a $d$-dimensional data record $U^i=[u_1^i ,u_2^i, \cdots u_d^i]$, and the perturbation process is conducted on each dimension with the privacy budget $\epsilon'=\epsilon/d$.

If $u_j^i$ is the value of a two-valued attribute, it is randomly flipped according to the following rule in RR:
\begin{small}
\begin{equation}\nonumber
u_j^i=\left\{
             \begin{array}{lr}
             u_j^i,&\mbox{with probability of $q=\frac{e^{\epsilon'}}{1+e^{\epsilon'}}$} \\
            
             1-u_j^i, &\mbox{with probability of $\frac{1}{1+e^{\epsilon'}}$} 
             \end{array}
\right.
\end{equation}
\end{small}

If $u_j^i$ is the value of a multivalued attribute with $s$ possible values $c_1, c_2, \cdots c_s$, it is randomly flipped according to the following rule in RR:
\begin{small}
\begin{equation}\nonumber
u_j^i=\left\{
             \begin{array}{lr}
             u_j^i,& \mbox{with probability of $q_{ss}=\frac{e^{\epsilon_i}}{s-1+e^{\epsilon'}}$}\\ [2mm]
  
             c_k\neq u_j^i, &\mbox{with probability of $\frac{1}{s-1+e^{\epsilon'}}$} 
             \end{array}
\right.
\end{equation}
of which $k=1,2,\cdots,s$.
\end{small}

After receiving the noisy data from each user, the server computes the marginal probability distribution $\vec{\lambda}$, estimates the original distribution $\hat{\vec{\pi}}$, and then calculates $Q$  for each attribute variable according to the methods presented in Sections \ref{sec:two-valued} and \ref{sec:multivalued}. Then it constructs a Bayesian network and conducts attribute clustering on the aggregated data to learn the correlations between different attribute variables. The processes of Bayesian network construction and attribute clustering are similar to those in Sections \ref{bay} and  \ref{att} except for a few minor changes: replace Line $6$ of Algorithm \ref{ag1} with a procedure that selects $(A_i,\Pi_i)$ with the largest $I(A,\Pi)$, since the process of the Bayesian network construction does not need to satisfy differential privacy, as the aggregated data at the server is already differentially private (guaranteed by local differential privacy). To further improve accuracy, Line $5$ of Algorithm \ref{ag2} can be replaced by ``Select the attribute $x$ with the maximal entropy in $S$''.

Next the server calculates $\tilde{Q}$ for each attribute variable based on $\hat{\vec{\pi}}$ and $Q$.
Then the server conducts the randomized perturbation by applying $\tilde{Q}$ on the aggregated data  to achieve invariant PRAM. According to the attribute clustering, each cluster may include more than one attribute variable which are strongly correlated. Thus we compute the transition probability matrix of the compound variable following the procedure presented in Section \ref{compound}.


\begin{theorem}
The DP2-Pub with a semi-honest server satisfies $\epsilon$-local differential privacy.
\end{theorem}

\begin{proof}
Each user perturbs its data record individually with the help of random response to get the privatized data, which provides local differential privacy. The operations of the server are all conducted on the privacy-preserved data, and the PRAM perturbation can be considered as a randomized mapping (post processing) without breaking differential privacy. Therefore, the DP2-Pub mechanism with a semi-honest server is differentially private with privacy budget $\epsilon$.
\end{proof}

\section{Experimental Evaluations}\label{sec:eva}

In this section, we conduct extensive experiments to demonstrate the performance of our DP2-Pub mechanism and compare it with two benchmark approaches, PrivBayes \cite{ZCPC} and DPPro \cite{XRZ}, on four real-world datasets of NLTCS \cite{nltcs}, ACS \cite{acs}, BR2000 \cite{acs} and Adult \cite{adult}. The data utility is evaluated and analyzed from two aspects, namely the total variation distance between the original dataset and the perturbed dataset, and the classification error rate of the SVM classification on the perturbed datasets.

\subsection{Experimental Settings}\label{sec:datasets}

\subsubsection{Datasets} We make use of four real-world datasets in our experiments: NLTCS \cite{nltcs} consists of records of $21574$ individuals participated in the National Long Term Care Survey, and each record has 16 attributes; ACS \cite{acs} includes $47461$ records of personal information from the $2013$ and $2014$ ACS sample sets in IPUMS-USA, where each record has $23$ attributes; BR2000 \cite{acs} consists of $38000$ census records with $14$ attributes collected from Brazil in the year 2000; and Adult \cite{adult} contains personal information such as gender, salary, and education level of $45222$ records extracted from the 1994 US Census, where each record has 15 attributes. The first two datasets only contain binary attribute values while the last two possess continuous as well as categorical attributes with multiple values. We summarize the statistics of these datasets in Table \ref{tab:data}. 

\begin{table}[!htb]
\renewcommand{\arraystretch}{1.0}
\caption{Data Statistics} 
\centering
\scalebox{0.8}{
\begin{tabular}{|c|c|c|c|}
\hline
\textbf{Dataset}& \textbf{Cardinality}   &    \textbf{Dimensionality}         &   \textbf{Domain size} \\    
\hline 
\hline                    
NLTCS &  $21574$ &  $16$ & $\approx 2^{16}$   \\
\hline
ACS &  $ 47461$ &  $23$ & $\approx 2^{23}$ \\
\hline
BR2000 &  $ 38000$ &  $14$  & $\approx 2^{32} $ \\
\hline
Adult &  $ 45222$ &  $15$ & $\approx 2^{52}$ \\                      
\hline
\end{tabular}}
\label{tab:data}
\end{table}

\subsubsection{Evaluation Metrics} We consider two tasks to evaluate the performance of DP2-Pub. The first task is to study the accuracy of $\alpha$-way marginals of the perturbed datasets. We evaluate the $\alpha$-way marginals of the four datasets by adopting the \emph{total variation distance} \cite{TSAB}  between the noisy marginal distribution and that of the original datasets, which is shown in Eq.~\eqref{tvd}. 
\begin{small}
\begin{equation}\label{tvd}
\begin{aligned}
AVD(X,Z)=\frac{1}{2}\|P_X-P_Z\|_1=\frac{1}{2}\sum_{w\in \Omega}|P_X(w)-P_Z(w)|
\end{aligned}
\end{equation}
\end{small}
where $\Omega$ is the domain of the probability variable $X$ and $Z$; $P_X$ and $P_Z$ are the probability distributions of the original attribute variable $X$ and the perturbed one $Z$, respectively.

Then we compute the average results of the total variation distance over all $\alpha$-way marginals as the final result -- a lower distance implies a better utility. More specifically, in our experiments, we evaluate the $3$-way and $4$-way marginals on binary datasets NLTCS and ACS, and $2$-way and $3$-way marginals on BR2000 and Adult, since the domain size of BR2000 and Adult are prohibitively large leading to very complex joint distributions.

The second task is to evaluate the classification results of SVM classifiers. The purpose of data publication is to conduct data analysis and data mining. We adopt SVM to evaluate the data utility from the perspective of data applications, as SVM is the most popular classification approach among various data mining techniques with powerful discriminative features both in linear and non-linear classifications \cite{SNSR}. Specifically, we train two classifiers on ACS to predict whether an individual: (1) goes to school, (2) lives in a multi-generation family; four classifiers are constructed on NLTCS to predict whether an individual: (1) is unable to go outside, (2) is unable to manage money, (3) is unable to bathe, and (4) is unable to travel; two classifiers are trained on BR2000 to predict whether an individual (1) owns a private dwelling, (2) is a Catholic; and two classifiers are trained on Adult to predict whether an individual (1) is a female, (2) makes over $50$K a year. For each classifier, we use $80\%$ of the tuples of the dataset for training and the other $20\%$ as the testing set. The prediction accuracy of each SVM classifier is measured by the \emph{misclassification rate} on the testing set.

\subsubsection{Comparison Approaches}

For the two evaluation metrics mentioned above, we compare our mechanism DP2-Pub with two existing approaches: (1) PrivBayes \cite{ZCPC}, which first constructs a Bayesian network to model the correlations among the attributes in a dataset, then injects noise into each marginal distribution in the Bayesian network to realize differential privacy, and finally constructs an approximation to the data distribution of the original dataset using the Bayesian network and the noisy marginal distributions; (2) DPPro \cite{XRZ}, which projects a $d$-dimensional vector representation of a user's attributes into a lower $d$-dimensional space by a random projection, and then adds noise to each resultant vector. Note that we choose PrivBayes and DPPro for our comparison study because the former is a benchmark solution in a way of decomposing high-dimensional data into a set of low-dimensional marginal distributions while the latter is an effective approach of random projection.

\subsubsection{Parameter Settings} In our experiments, we use DP-Pub$^{1}$ to denote the case with a trusted server and DP-Pub$^{2}$ the one with a semi-honest server. The privacy budget $\epsilon$ of DP-Pub$^{1}$ is evenly distributed to the two phases, i.e., $\epsilon_1=\epsilon_2=\frac{1}{2}\epsilon$. For DP-Pub$^{2}$, there is no need to partition the privacy budget since the data is first locally differentially privatized, i.e., the privacy budget $\epsilon$ is completely allocated to the local privacy procedure. For the parameter $k$ used in the construction of the Bayesian network, we test $k=1,2,3$. Since the time cost for larger $k$ values is typically higher, we do not try the cases of $k>3$. Based on our experiments, we observe that the influence of $k$ on the experimental results is not obvious. The reason possibly lies in that the structure of the Markov blanket can help to accurately learn the data correlations between different attributes. In the following section, we present the experimental results of $k=2$.

\subsection{Experimental Results}\label{sec:results}

In this subsection, we carry out $50$ independent runs for each of the experiments mentioned above and report the averaged results for statistical confidence. 

\subsubsection{Results on Average Variation Distance}

For the task of examining the accuracy of $\alpha$-way marginals, we compute all the $\alpha$-dimensional attribute unions and compare the averaged variation distance of PrivBayes, DPPro, DP-Pub$^{1}$ and DP-Pub$^{2}$, with a varying privacy budget $\epsilon$ from $0.2$ to $1.6$.
\begin{figure}[!htb]
\begin{minipage}[t]{0.24\textwidth}
\centering
\includegraphics[width=\textwidth]{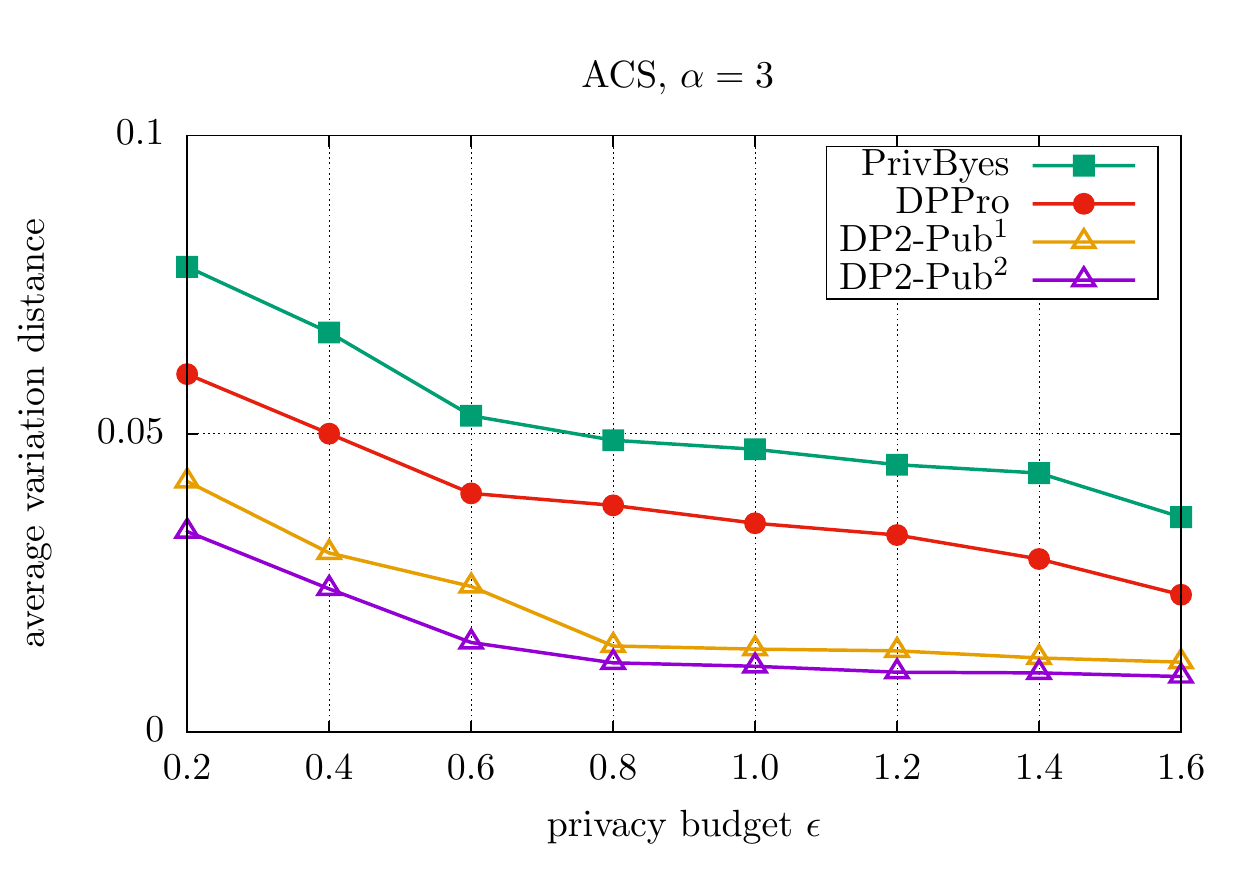}
\end{minipage}
\begin{minipage}[t]{0.24\textwidth}
\centering
\includegraphics[width=\textwidth]{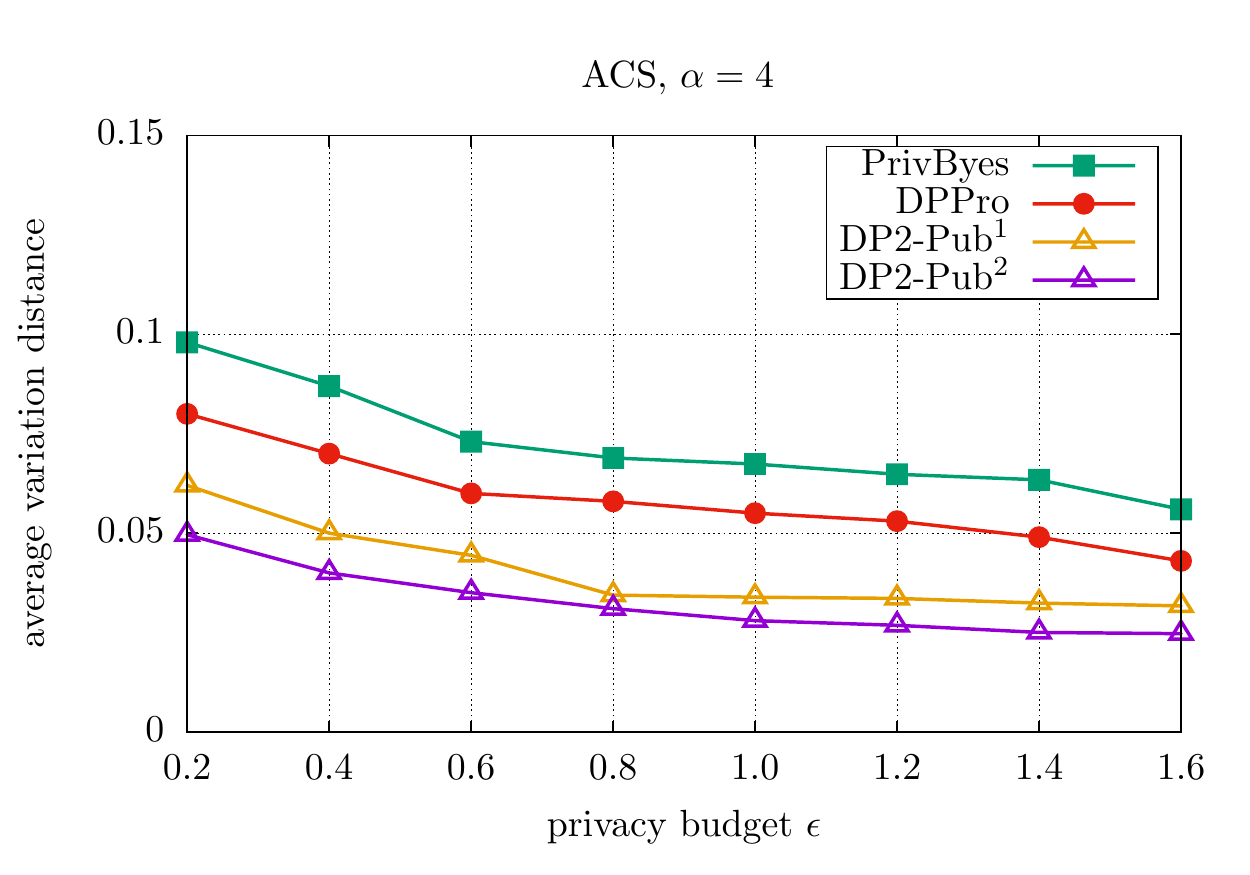}
\end{minipage}
\label{fig1}

\begin{minipage}[t]{0.24\textwidth}
\centering
\includegraphics[width=\textwidth]{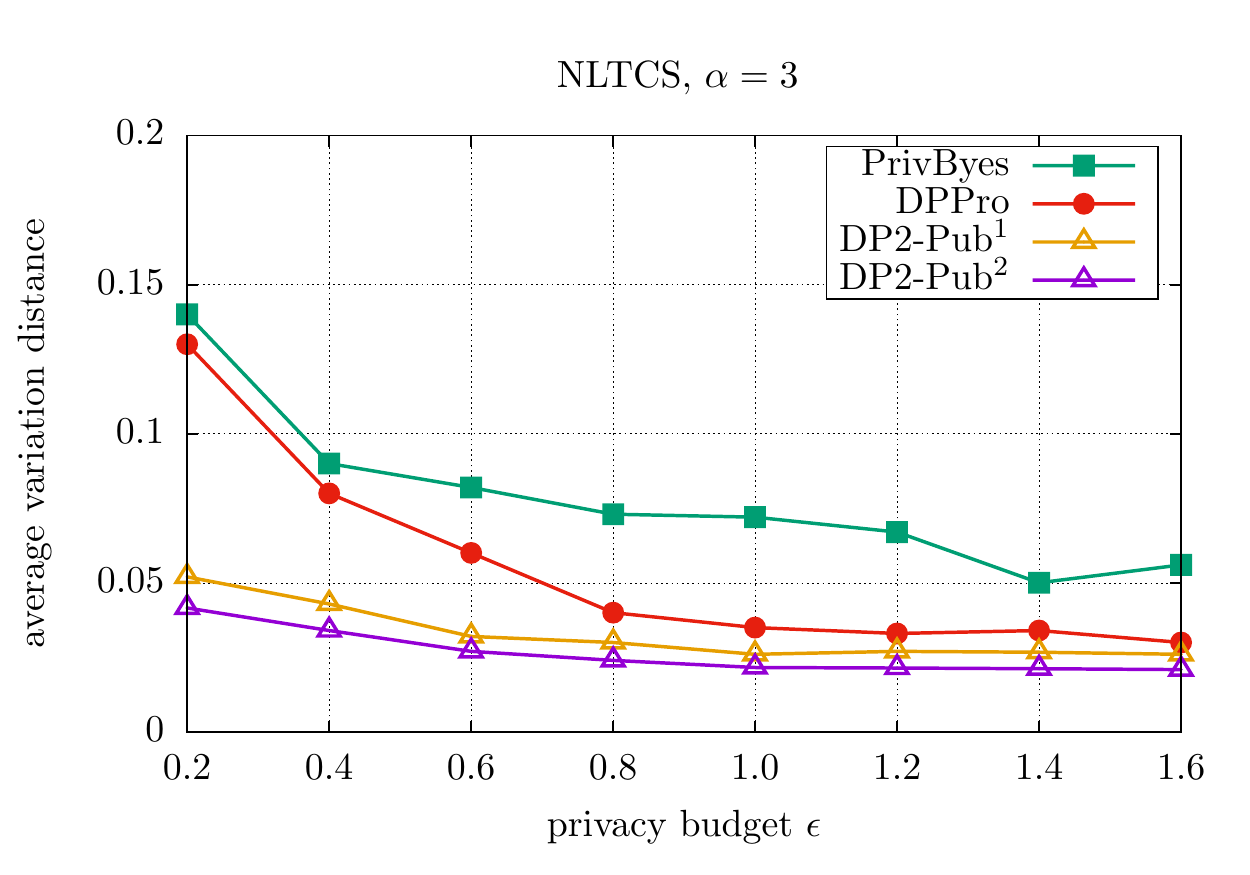}
\end{minipage}
\begin{minipage}[t]{0.24\textwidth}
\centering
\includegraphics[width=\textwidth]{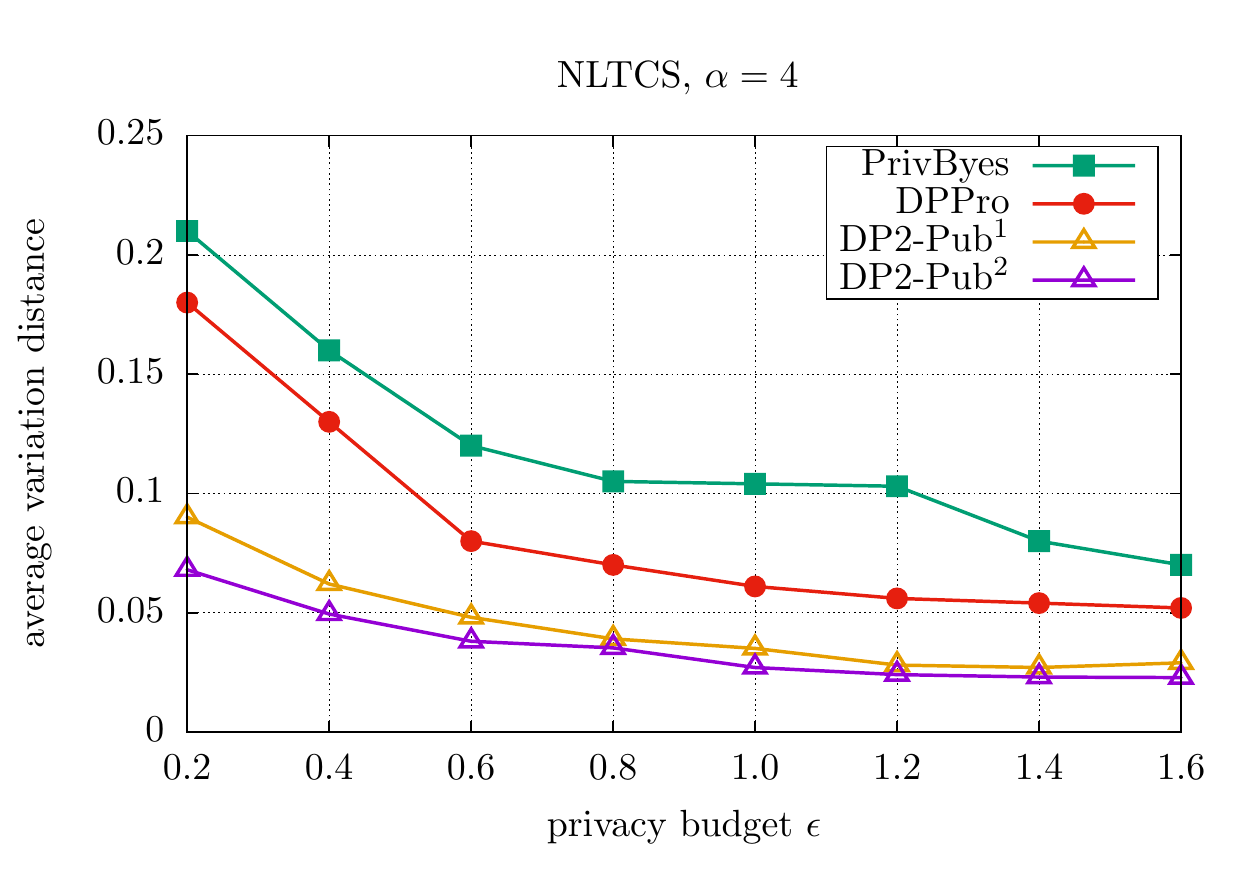}
\end{minipage}

\begin{minipage}[t]{0.24\textwidth}
\centering
\includegraphics[width=\textwidth]{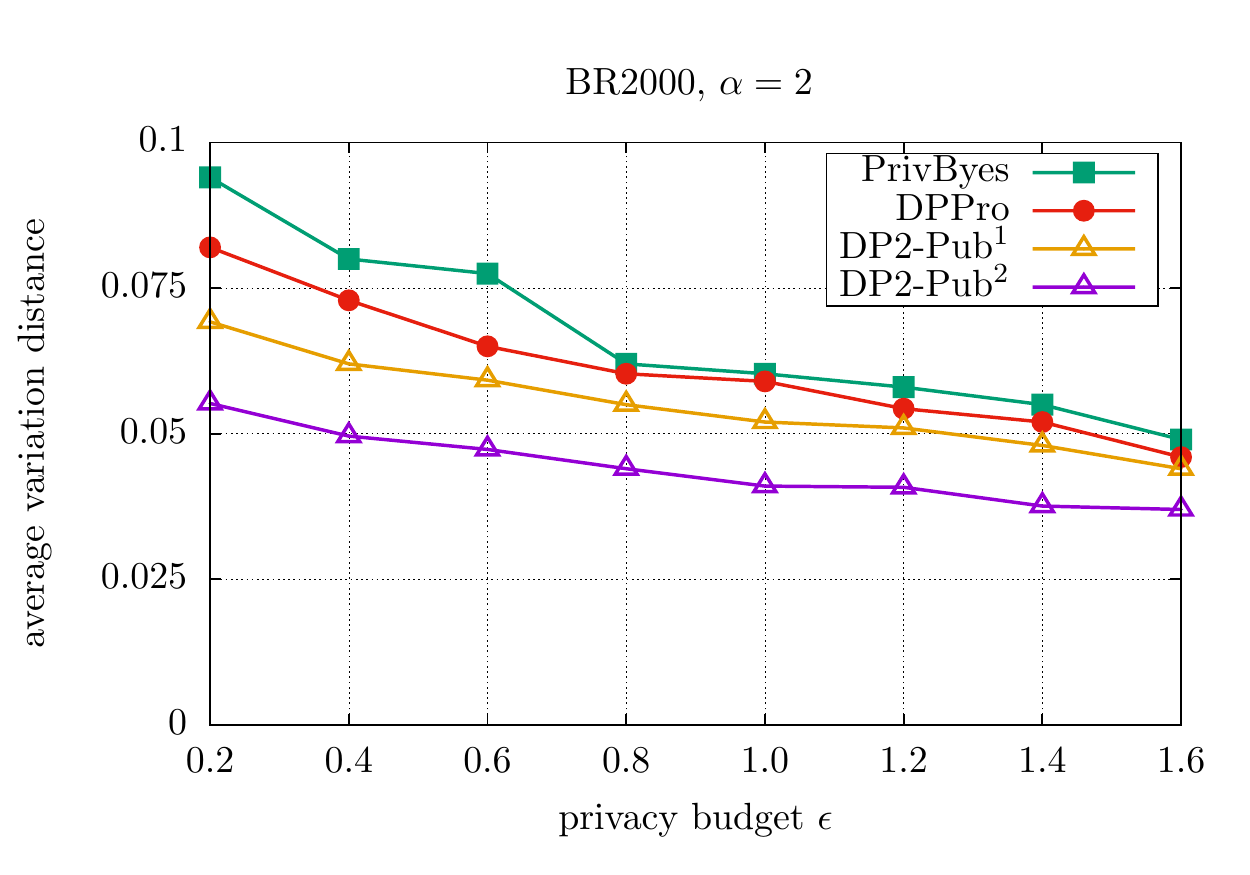}
\end{minipage}
\begin{minipage}[t]{0.24\textwidth}
\centering
\includegraphics[width=\textwidth]{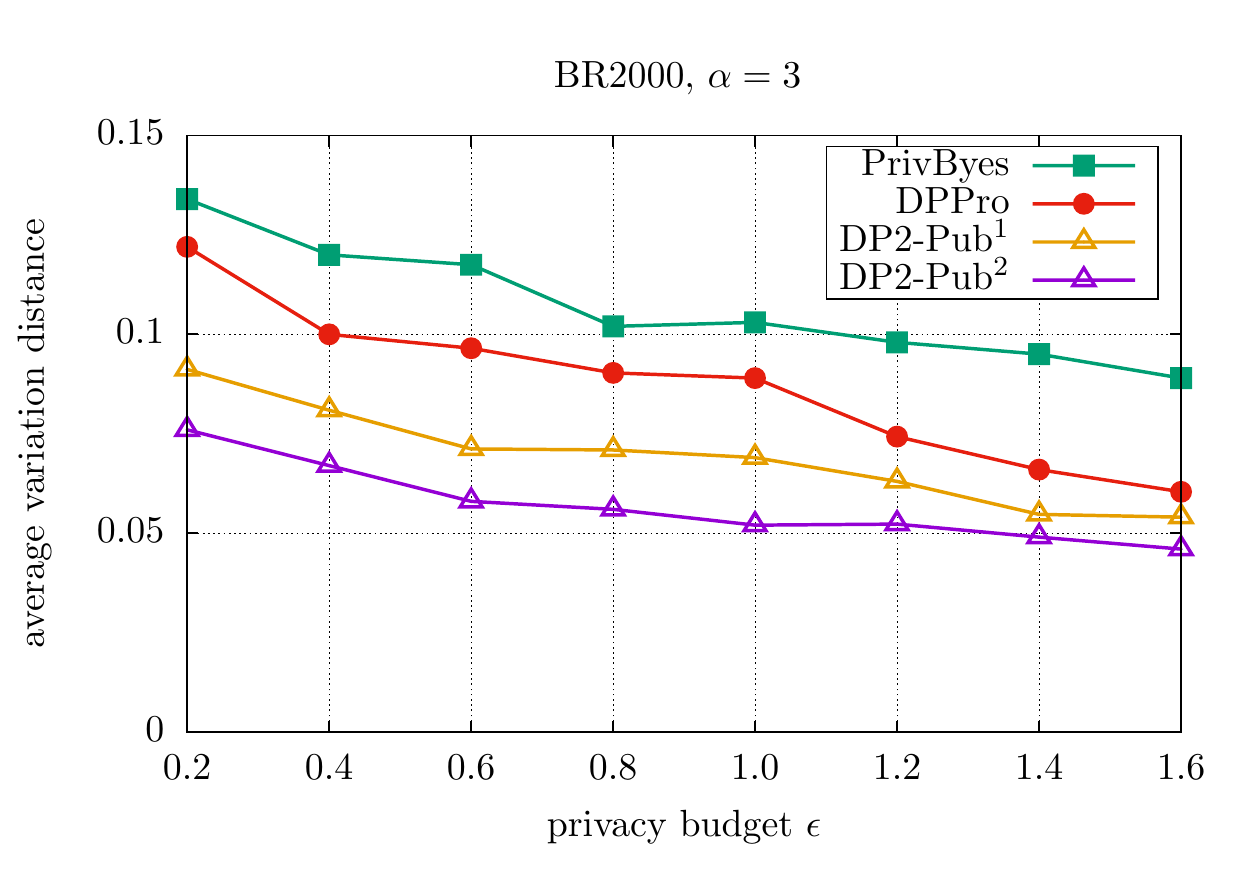}
\end{minipage}

\begin{minipage}[t]{0.24\textwidth}
\centering
\includegraphics[width=\textwidth]{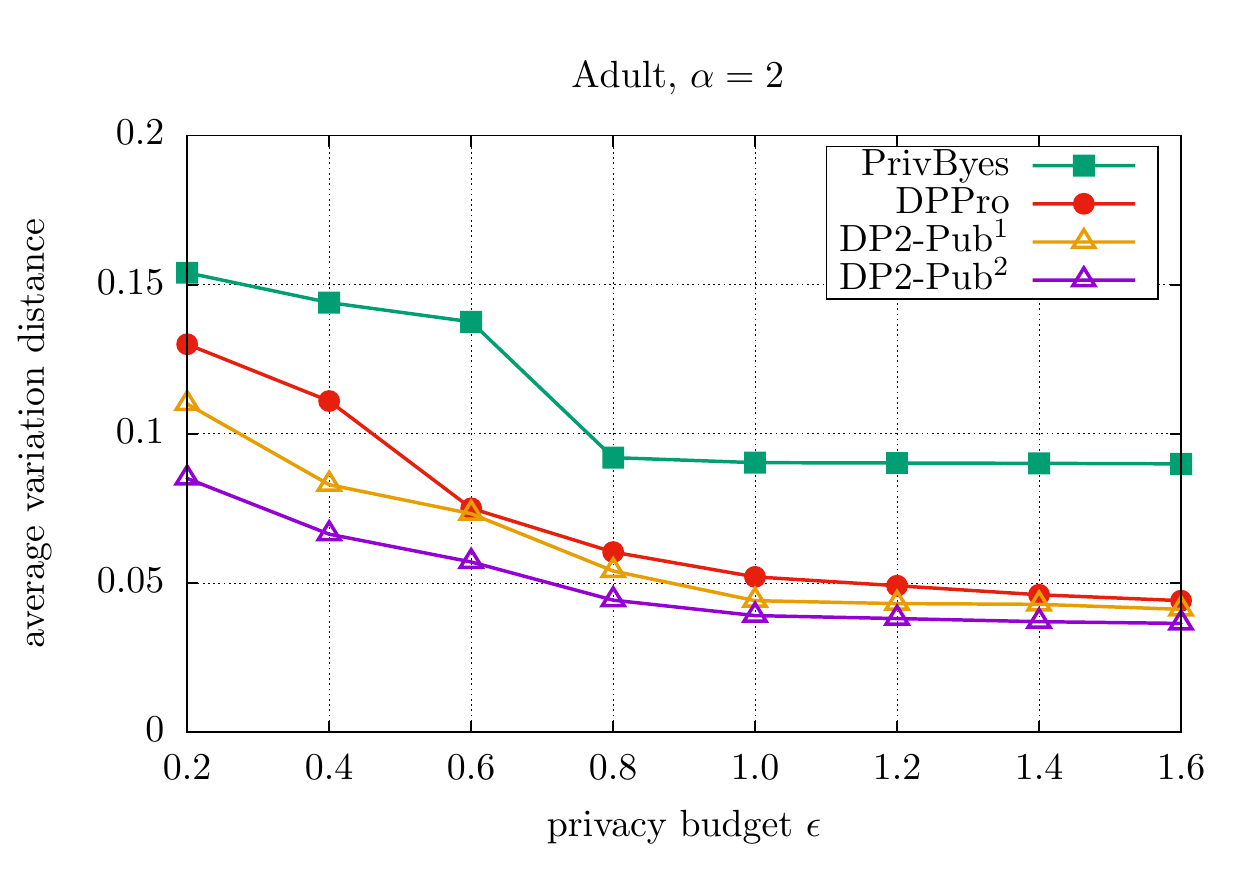}
\end{minipage}
\begin{minipage}[t]{0.24\textwidth}
\centering
\includegraphics[width=\textwidth]{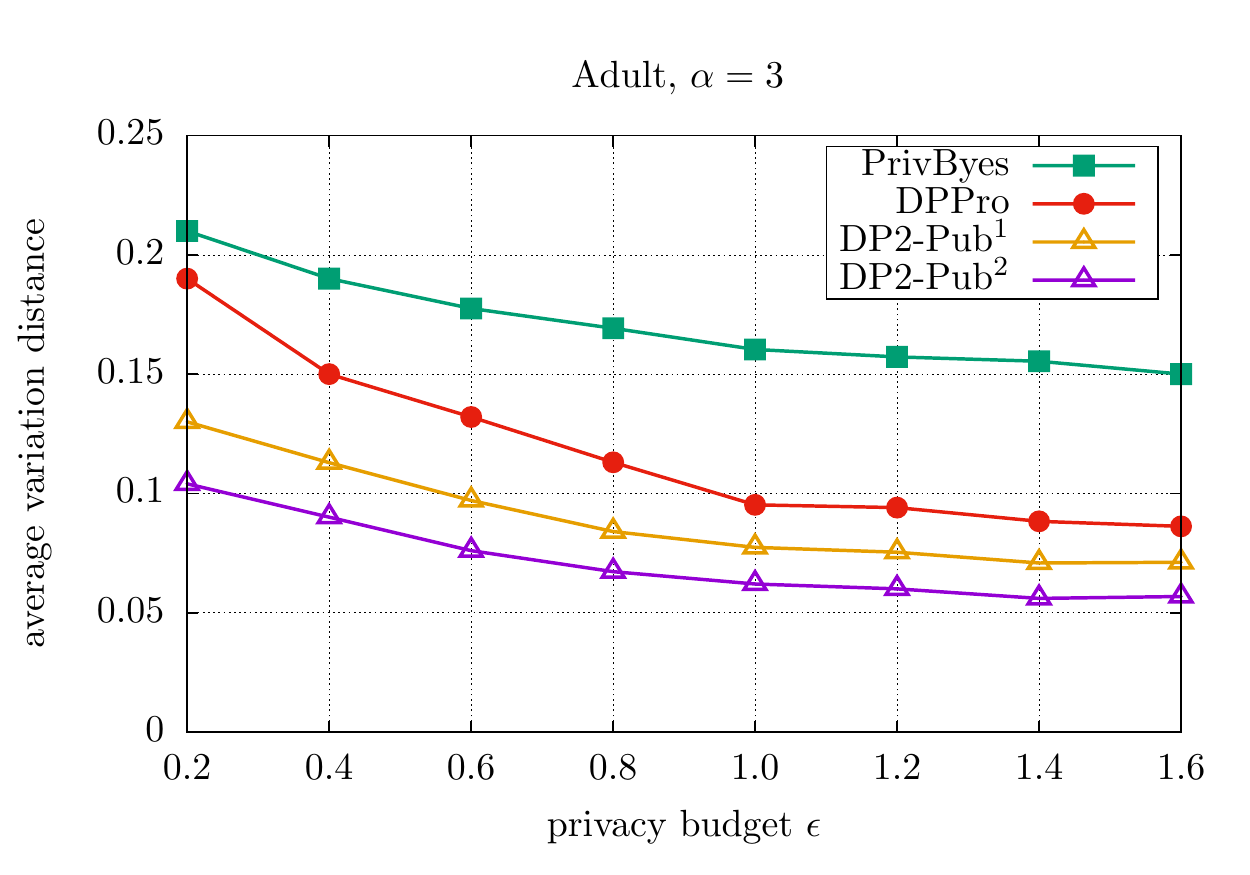}
\end{minipage}
\caption{Results of $\alpha$-way marginals with different $\epsilon$.}
\label{fig4}
\end{figure}

Figure \ref{fig4} shows the average results of the variation distance of each approach on the four datasets. From Figure \ref{fig4}, one can see that the average variation distances of these three approaches decrease when $\epsilon$ increases over the four datasets. It is obvious that when $\epsilon$ is larger, smaller noise is required, and the data utility is higher. One can also observe that our approach clearly outperforms PrivBayes and DPPro in all cases for ACS and NLTCS, while for BR2000 and Adult, the relative superiority is more pronounced when $\epsilon$ is small. There are several reasons that DP2-Pub outperforms PrivBayes and DPPro. First, PrivBayes constructs a Bayesian network to model the data correlation and generates a set of noisy conditional distributions of the original dataset. That is, for each attribute-parent pair in the Bayesian network, PrivBayes generates differentially private conditional distributions by adding Laplace noise which makes the data utility of the dataset drastically decrease. In our approach, we only utilize the Bayesian network to learn the correlations between different attributes and adopt our proposed invariant post randomization to achieve data perturbation, which ensures that there is almost no loss of statistical information. The probability distribution of each attribute variation is basically unchanged after the double-perturbation. Second, the random projection method DPPro does not consider the data characteristics and only preserves the pairwise $L_2$ distance when generating the random projection matrix, thus it may lead to relatively low utility especially when there exist data correlations between different attributes. In our approach DP2-Pub, we learn the data correlations of the original dataset and consider the importance of different attributes when allocating the privacy budget. 

DP-Pub$^{2}$ performs better than DP-Pub$^{1}$ according to the results shown in Figure \ref{fig4}. This is counter-intuitive as centralized differential privacy usually performs better than local differential privacy because centralized differential privacy adds noise based on the sensitivity of a particular query function while in local differential privacy noise is added via post randomization. But in DP-Pub$^{1}$, noise is added for differentially private Bayesian network construction and for post randomization, with none of them considering the sensitivity of a particular query function, which is more general at the cost of lower utility. Moreover, at the same budget level, adding noise at two phases increases the total amount of noise as the added noise amount is not linearly proportional to the privacy budget -- it is super-linear, which also contributes to the lower utility of DP-Pub$^{1}$.


\subsubsection{Results on SVM classification}

\begin{figure}[!htp]
\begin{minipage}[t]{0.24\textwidth}
\centering
\includegraphics[width=\textwidth]{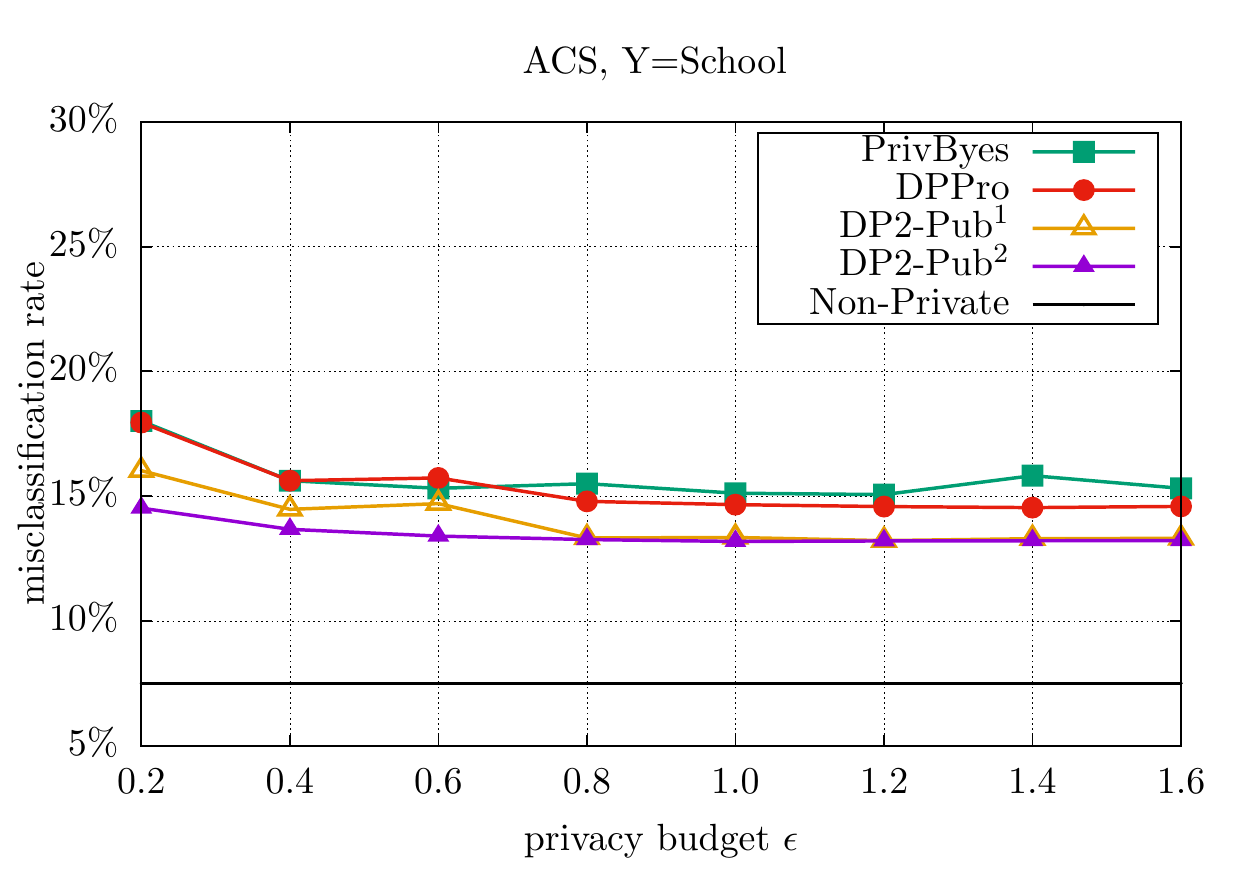}
\end{minipage}
\begin{minipage}[t]{0.24\textwidth}
\centering
\includegraphics[width=\textwidth]{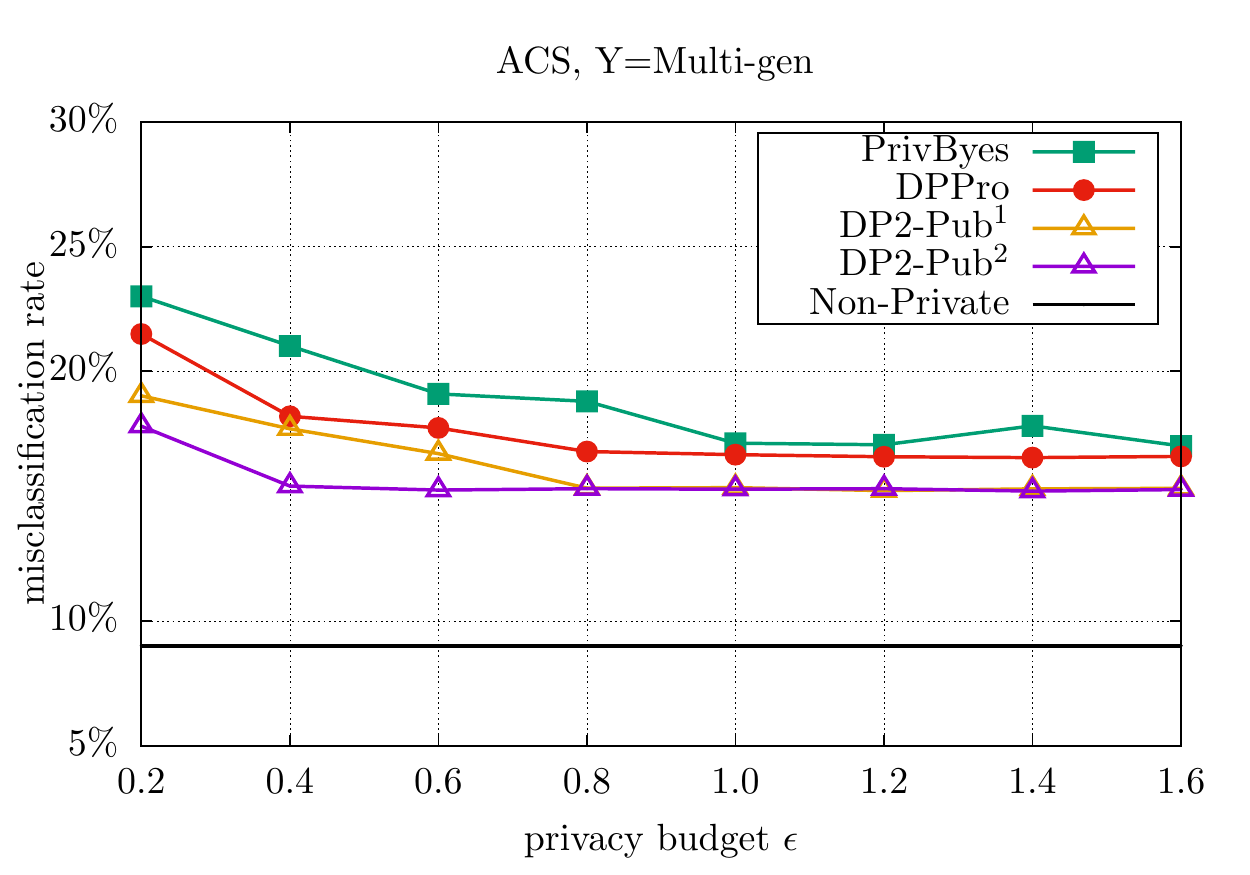}
\end{minipage}

\begin{minipage}[t]{0.24\textwidth}
\centering
\includegraphics[width=\textwidth]{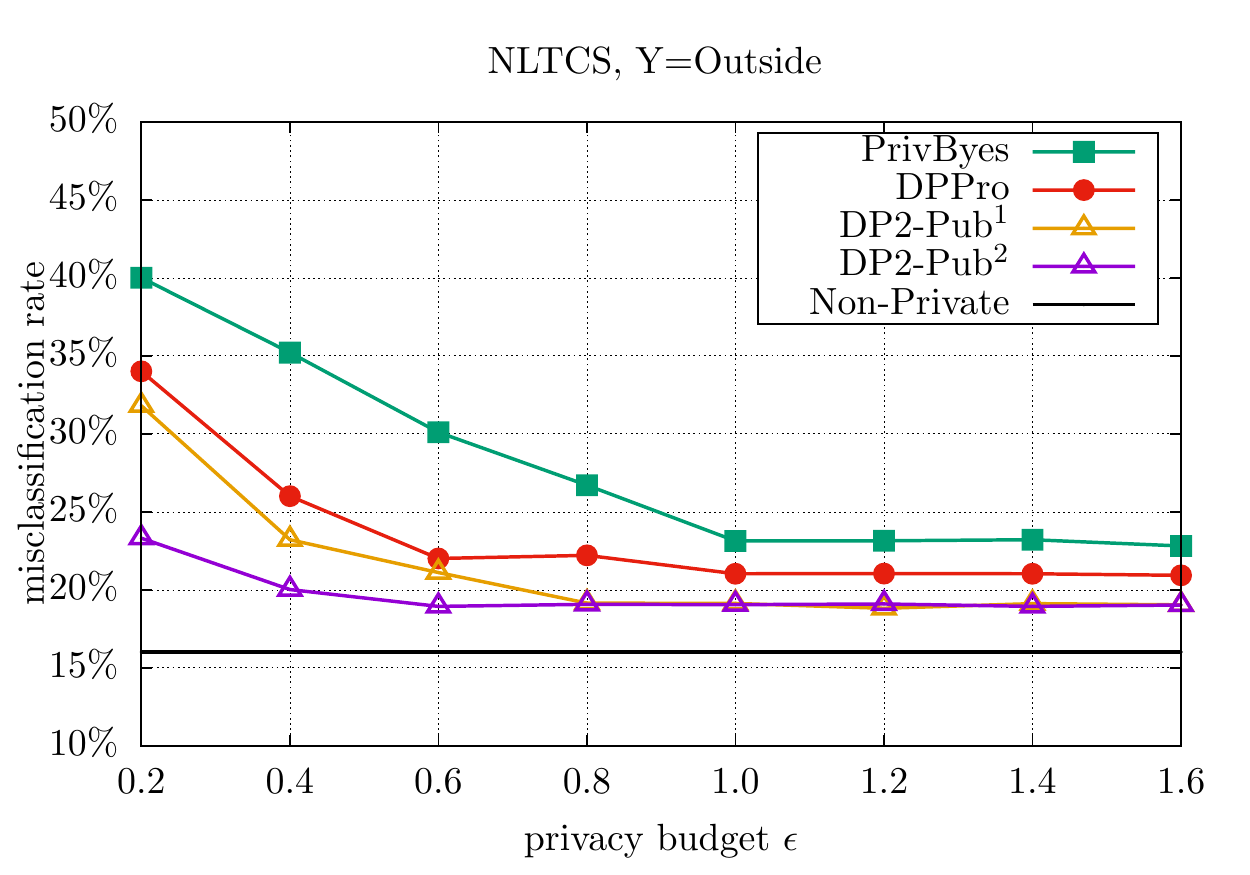}
\end{minipage}
\begin{minipage}[t]{0.24\textwidth}
\centering
\includegraphics[width=\textwidth]{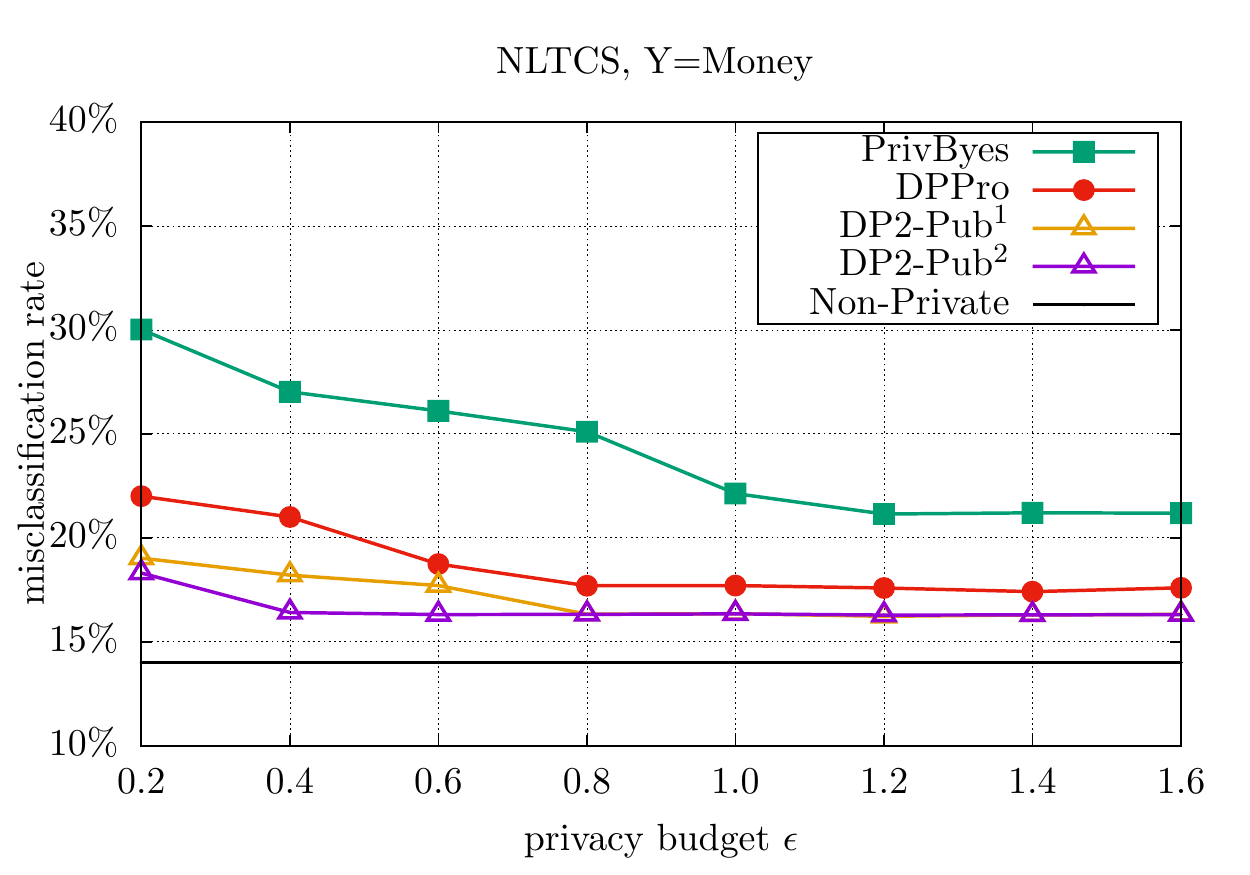}
\end{minipage}

\begin{minipage}[t]{0.24\textwidth}
\centering
\includegraphics[width=\textwidth]{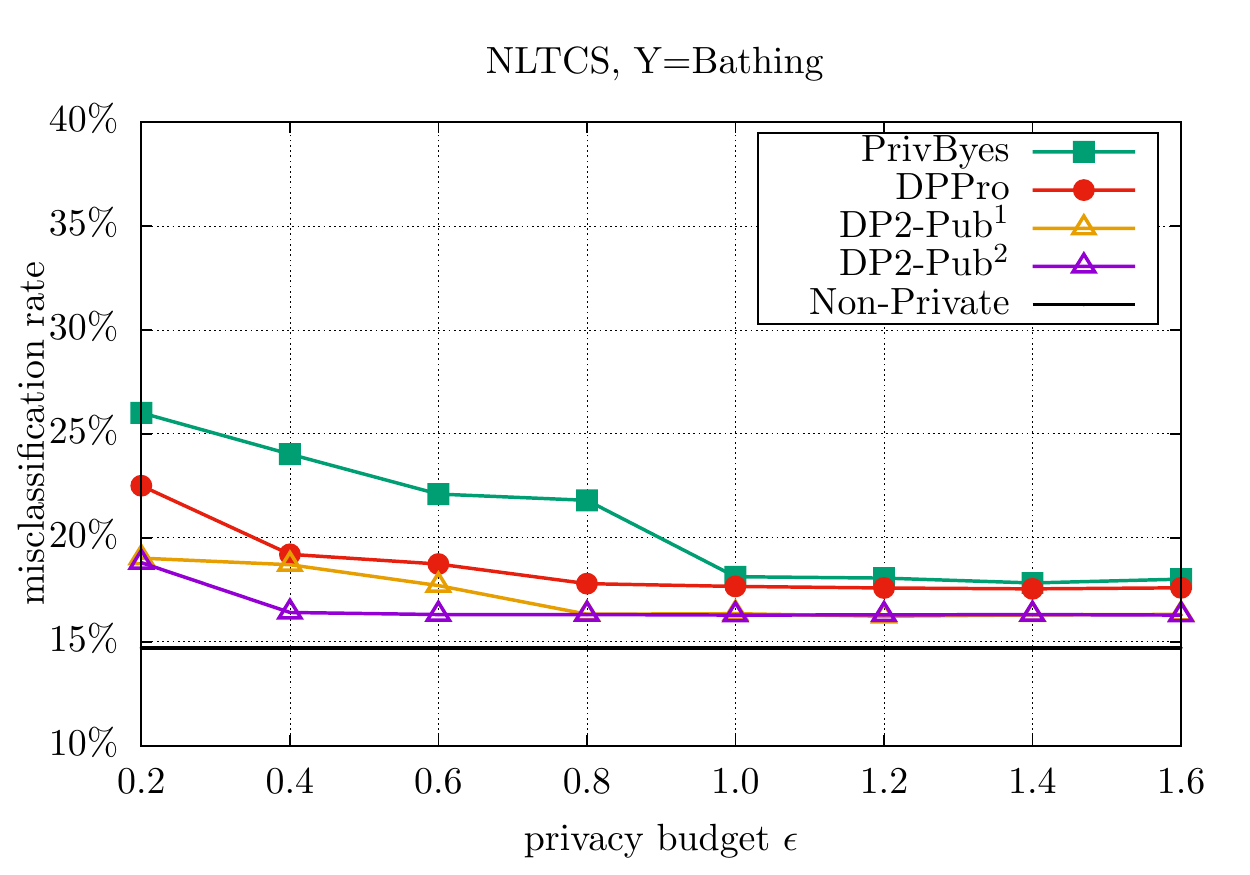}
\end{minipage}
\begin{minipage}[t]{0.24\textwidth}
\centering
\includegraphics[width=\textwidth]{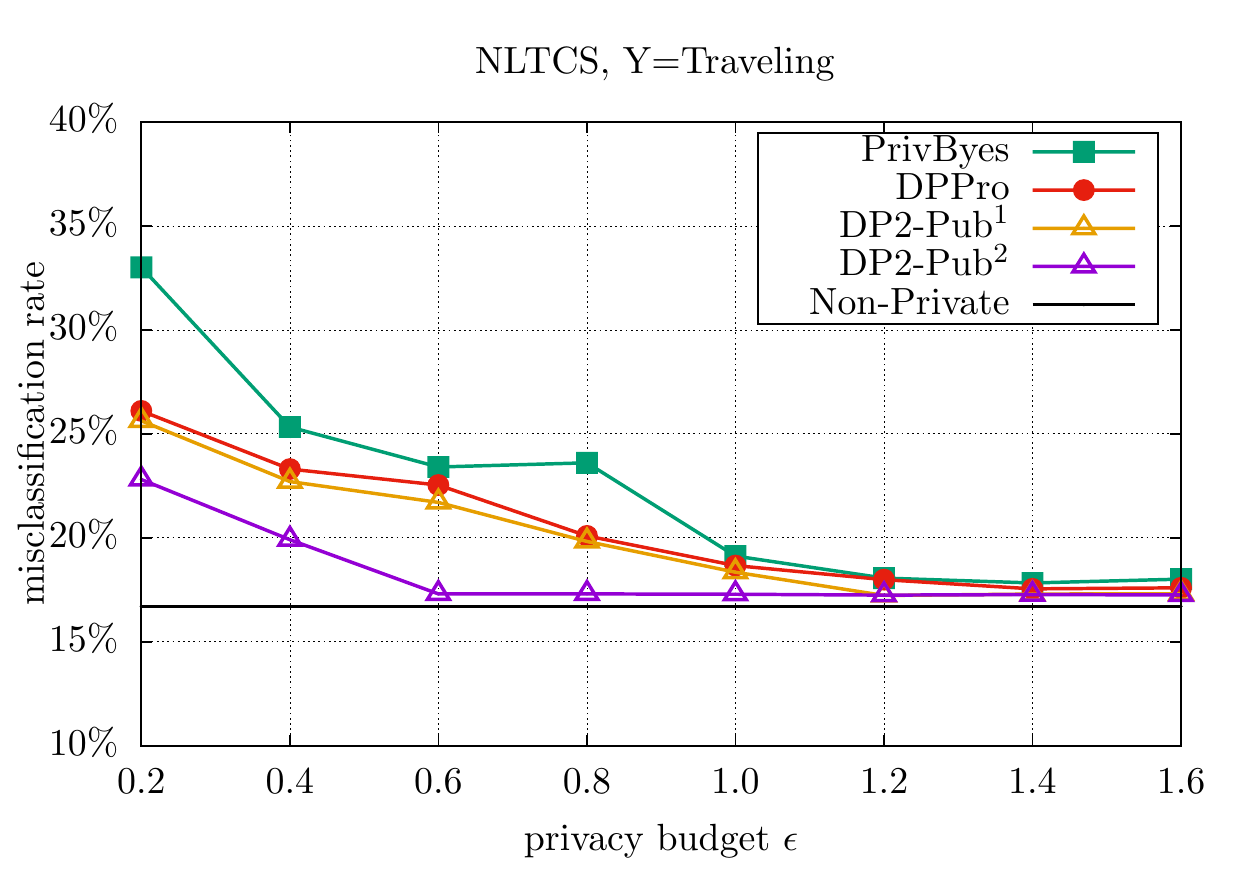}
\end{minipage}

\begin{minipage}[t]{0.24\textwidth}
\centering
\includegraphics[width=\textwidth]{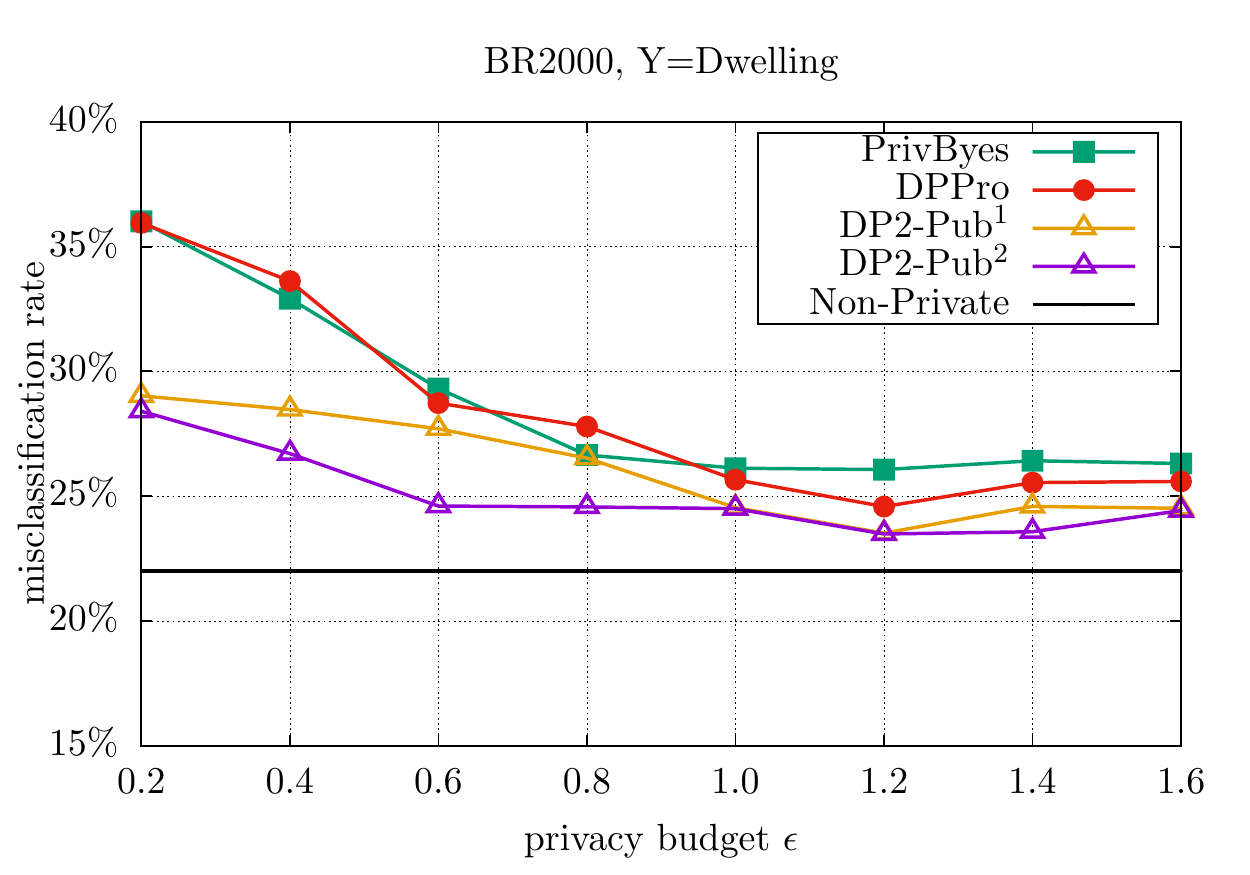}
\end{minipage}
\begin{minipage}[t]{0.24\textwidth}
\centering
\includegraphics[width=\textwidth]{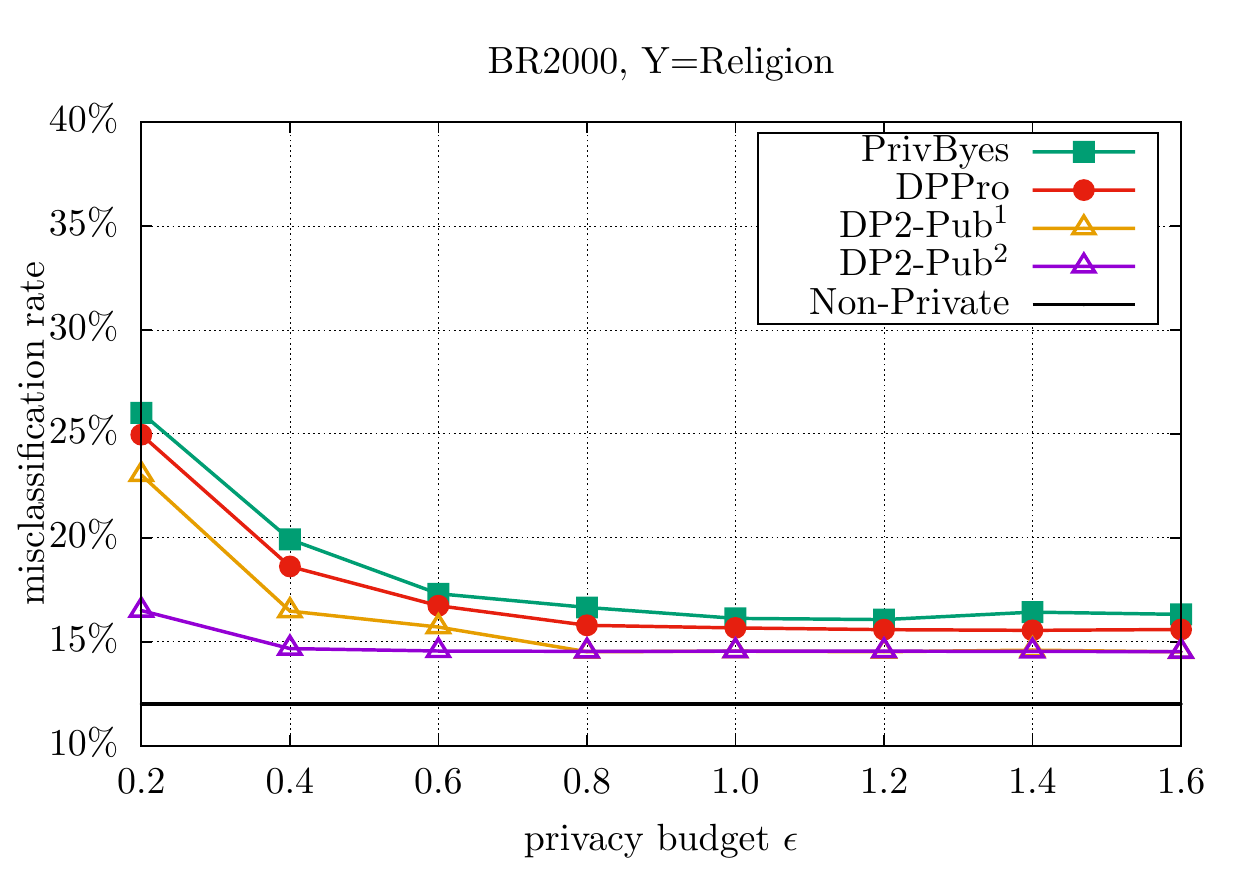}
\end{minipage}

\begin{minipage}[t]{0.24\textwidth}
\centering
\includegraphics[width=\textwidth]{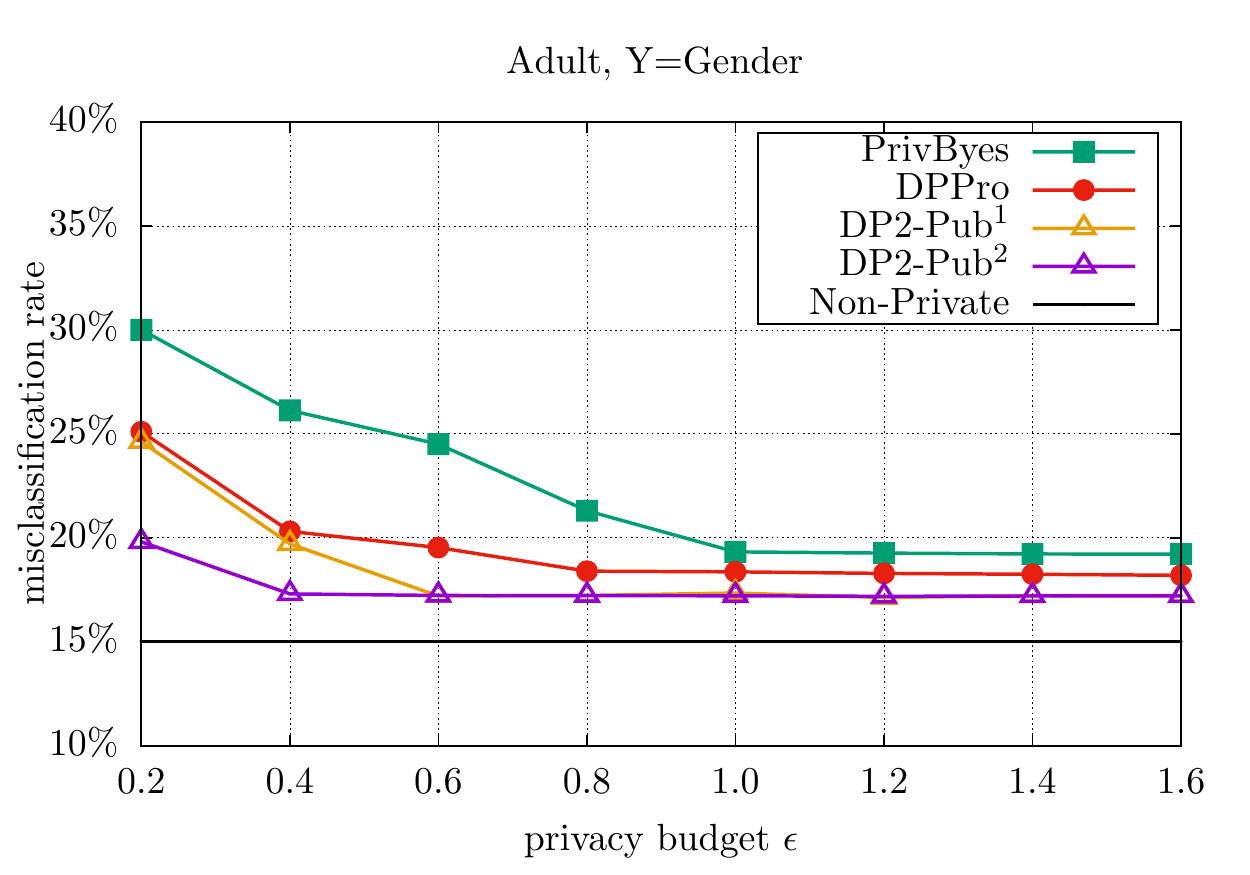}
\end{minipage}
\begin{minipage}[t]{0.24\textwidth}
\centering
\includegraphics[width=\textwidth]{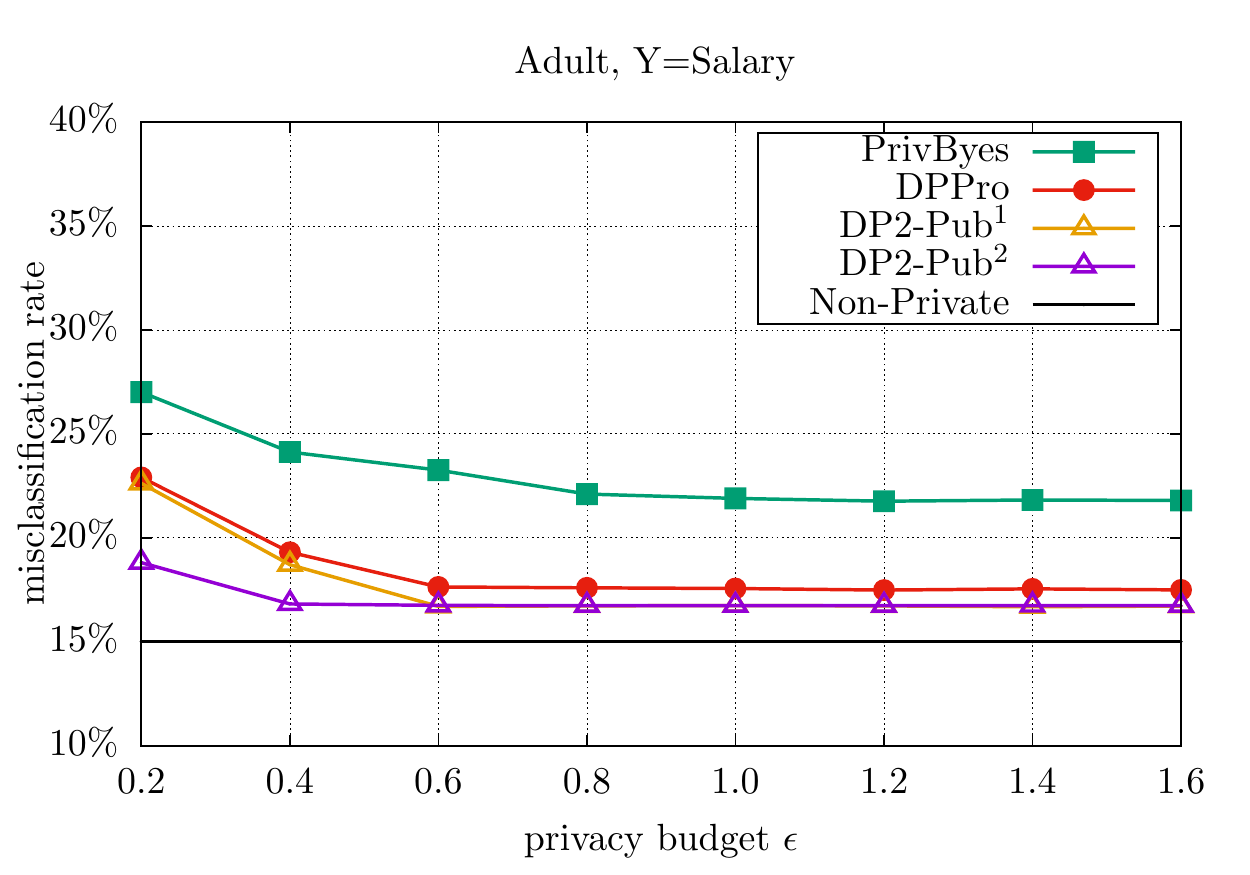}
\end{minipage}
\caption{Results of SVM with different $\epsilon$}
\label{fig5}
\end{figure}

For the second task, we evaluate the performance of PrivBayes, DPPro, DP2-Pub$^{1}, $DP-Pub$^{2}$, and Non-Private (no DP is considered) for SVM classification. Figure \ref{fig5} shows the misclassification rate of each approach under various different privacy budgets. One can see that the error of Non-Private remains unchanged for all $\epsilon$ since it does not consider differential privacy. One can also see that both DP-Pub$^{1}$ and DP-Pub$^{2}$ outperform PrivBayes and DPPro on almost all datasets. The reason for the higher classification accuracy of our approach lies in that it can achieve higher data utility with a better retention of correlations among attribute variables and a higher accuracy of joint distributions. More specifically, both DP-Pub$^{1}$ and DP-Pub$^{2}$ retain the data characteristics while satisfying privacy guarantee, thus can help to obtain good results of SVM classifications. Moreover, the misclassification rate decreases faster when $\epsilon$ increases from $0.2$ to $0.6$, and the decrease of the misclassification rate is not obvious when $\epsilon$ is larger than $0.8$. This indicates that a higher privacy level with a small $\epsilon$ leads to a lower data utility.

\section{Conclusions and Future Research}\label{sec:con}

In this paper, we propose a differentially private data publication mechanism DP2-Pub consisting of two phases, attribute clustering and data randomization. Specifically, in the first phase, we present the procedure of attribute clustering using the Markov blanket model based on the differentially private Bayesian network to achieve attribute clustering and obtain a reasonable allocation of privacy budget. In the second phase, we design an invariant post randomization method by conducting a double-perturbation while satisfying local differential privacy. Our privacy analysis shows that DP2-Pub satisfies differential privacy. We also extend our mechanism making it suitable for the scenario with a semi-honest server in a local-differential privacy manner. Comprehensive experiments on four real-world datasets demonstrate that DP2-Pub outperforms existing methods and improves data utility with strong privacy guarantee.

In our future research, we intend to combine other effective dimensionality reduction techniques\cite{OPM,LRS} with differential privacy to investigate their impact on the data utility of published data. Particularly, we intend to combine DP with manifold learning \cite{OPM}, which is a popular approach for non-linear dimensionality reduction that maps a high dimensional data space into a low-dimensional manifold representation of the data while preserving a certain form of geometric relationships between the data points.

\section*{Acknowledgment}

This work was partially supported by the US National Science Foundation under grant CNS-1704397. 

\footnotesize
\bibliographystyle{IEEEtran}      
\bibliography{ref}


\begin{IEEEbiography}[{\includegraphics[width=1in,height=1.25in,clip,keepaspectratio]{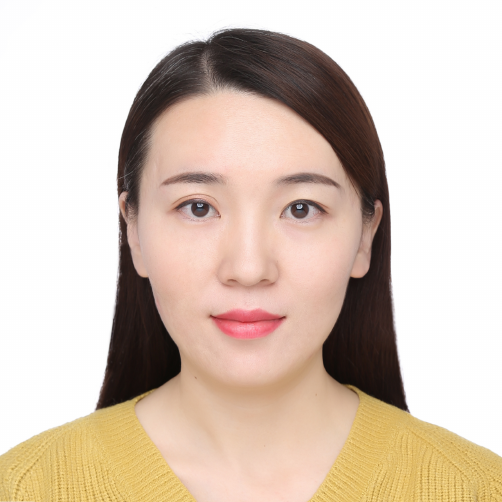}}]{Honglu Jiang}
received her Ph.D. degree in computer science from The George Washington University in 2021. Currently she is an assistant professor in the Department of Computer Science and Software Engineering at Miami University. Her research interests include wireless and mobile security, privacy preservation, federated learning and data analytics. She is a member of IEEE.
\end{IEEEbiography}


\begin{IEEEbiography}[{\includegraphics[width=1in,height=1.25in,clip,keepaspectratio]{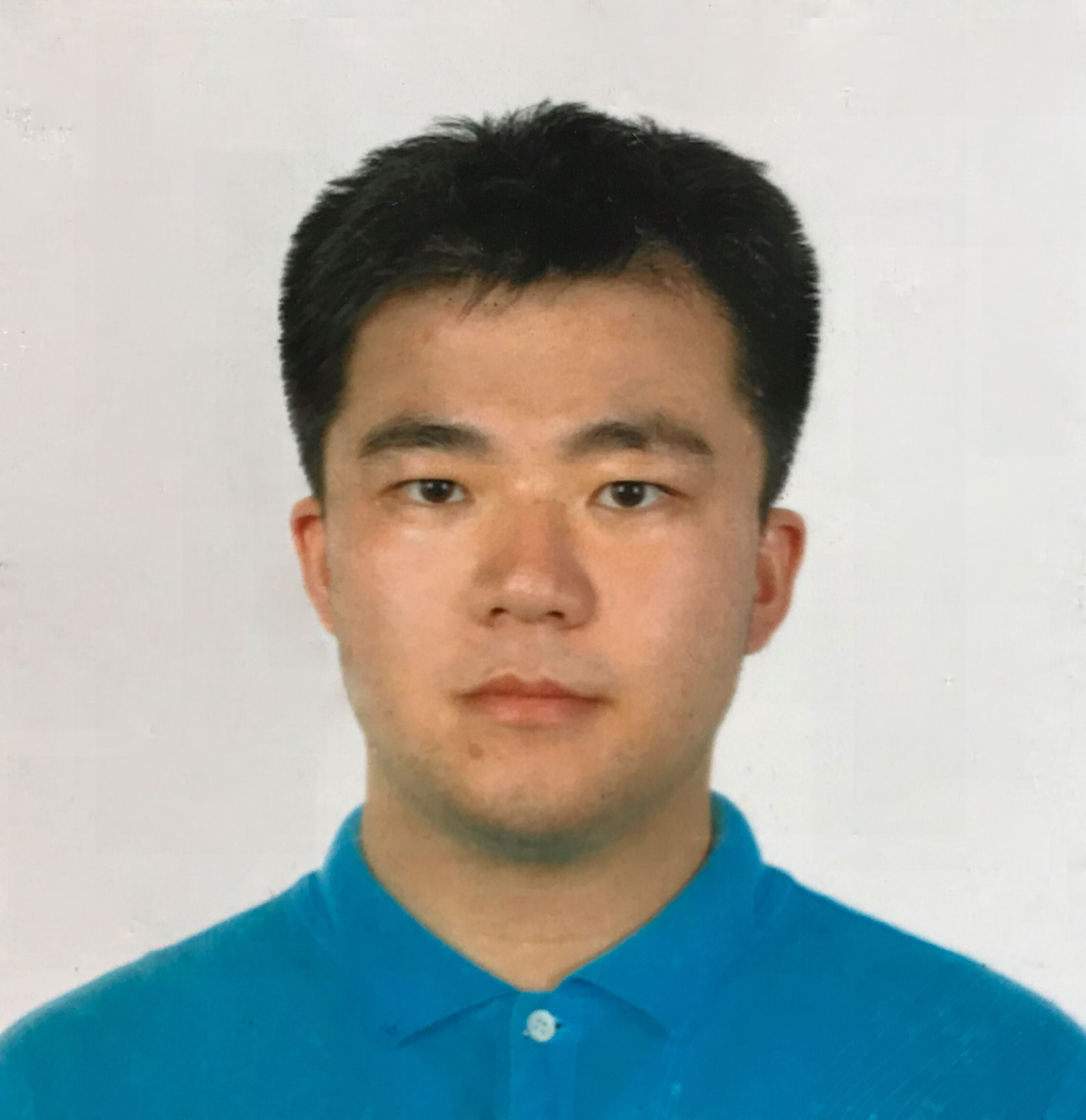}}]{Haotian Yu}
received his B.A. degree from the University of Minnesota in 2017, and M.S. degree in Data Analytics from The George Washington University in 2020. His research interests include data analysis for social networks.
\end{IEEEbiography}


\begin{IEEEbiography}[{\includegraphics[width=1in,height=1.25in,clip,keepaspectratio]{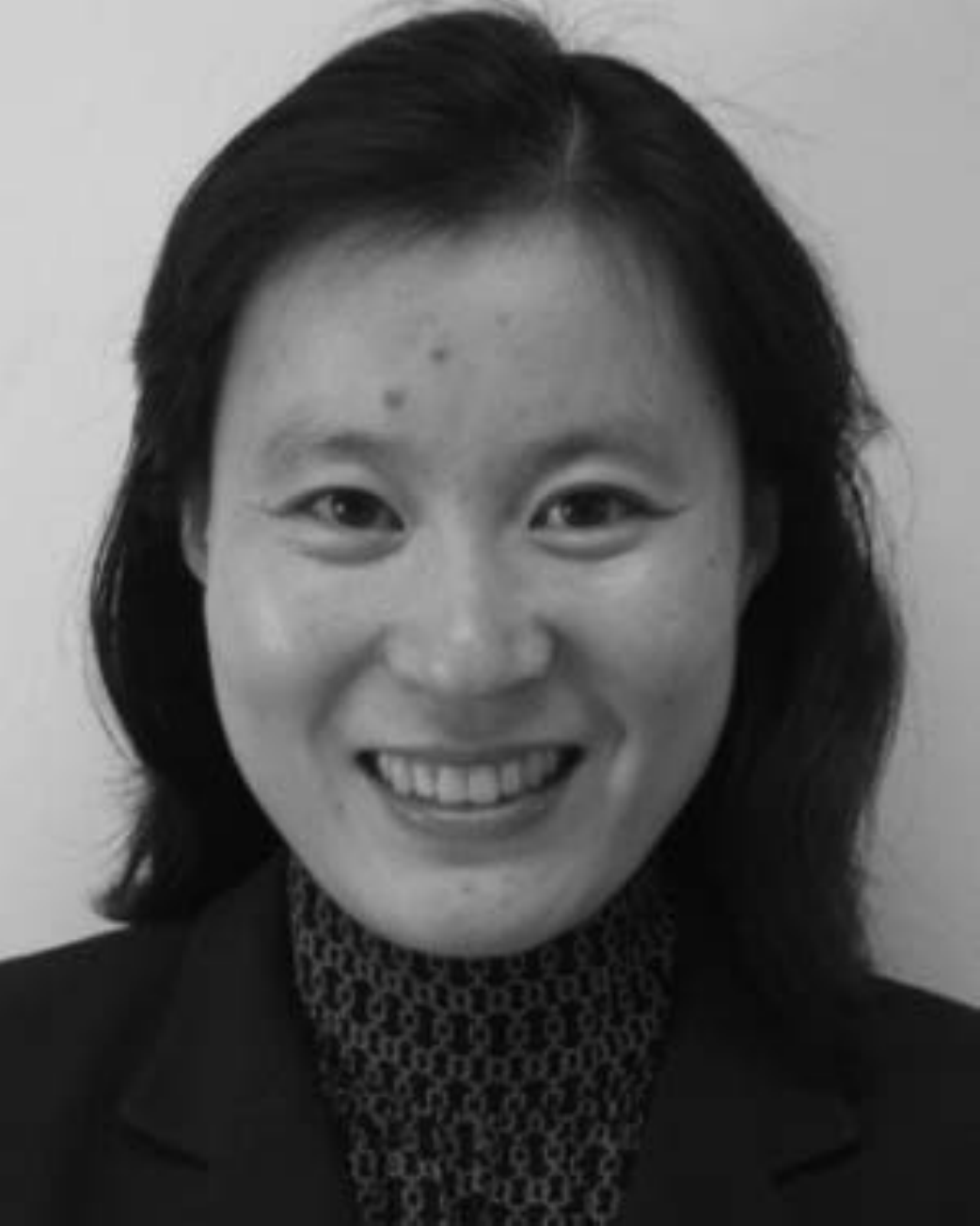}}]{Xiuzhen Cheng}
received her M.S. and Ph.D. degrees in computer science from the University of Minnesota---Twin Cities in 2000 and 2002, respectively. She is a professor of Computer Science at Shandong University, P. R. China. Her current research focuses on Blockchain computing, privacy-aware computing, and wireless and mobile security. She served/is serving on the editorial boards of several technical journals and the technical program committees of various professional conferences/workshops. She was a faculty member in the Department of Computer Science at The George Washington University from September 2002 to August 2020, and worked as a program director for the US National Science Foundation (NSF) from April to October in 2006 (full time) and from April 2008 to May 2010 (part time). She received the NSF CAREER Award in 2004. She is a member of ACM, and a Fellow of IEEE.
\end{IEEEbiography}


\begin{IEEEbiography}[{\includegraphics[width=1in,height=1.25in,clip,keepaspectratio]{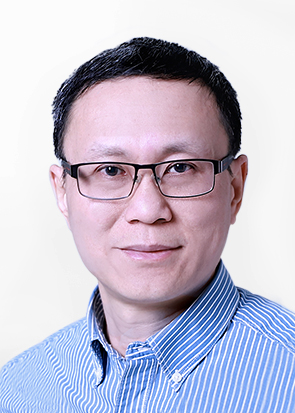}}]{Jian Pei} is currently Professor at Duke University.  His professional interest is to facilitate efficient, fair, and sustainable usage of data and data analytics for social, commercial and ecological good. Through inventing, implementing and deploying a series of data mining principles and methods, he produced remarkable values to academia and industry. His algorithms have been adopted by industry, open source toolkits and textbooks. His publications have been cited more than 107,000 times. He is also an active and productive volunteer for professional community services, such as chairing ACM SIGKDD, running many premier academic conferences in his areas, and being editor-in-chief or associate editor for the flagship journals in his fields. He is recognized as a fellow of the Royal Society of Canada (i.e., the national academy of Canada), the Canadian Academy of Engineering, ACM, and IEEE.  He received a series of prestigious awards, such as the ACM SIGKDD Innovation Award, the ACM SIGKDD Service Award, and the IEEE ICDM Research Award.
\end{IEEEbiography}


\begin{IEEEbiography}[{\includegraphics[width=1in,height=1.25in,clip,keepaspectratio]{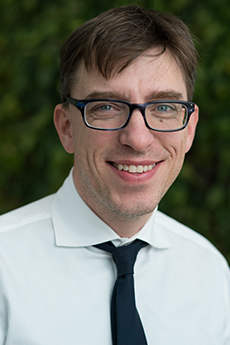}}]{Robert Pless} received the BS degree in computer science from Cornell University in 1994 and the PhD degree in computer science from the University of Maryland in 2000. He is a Patrick and Donna Martin Professor of Computer Science at The George Washington University. His research interests focus on computer vision with applications in environmental science, medical imaging, robotics and virtual reality. He chaired the IEEE Workshop on Omnidirectional Vision and Camera Networks (OMNIVIS) in 2003, and the MICCAI Workshop on Manifold Learning in Medical Imagery in 2008. He received the NSF career award in 2006.
\end{IEEEbiography}


\begin{IEEEbiography}[{\includegraphics[width=1in,height=1.25in,clip,keepaspectratio]{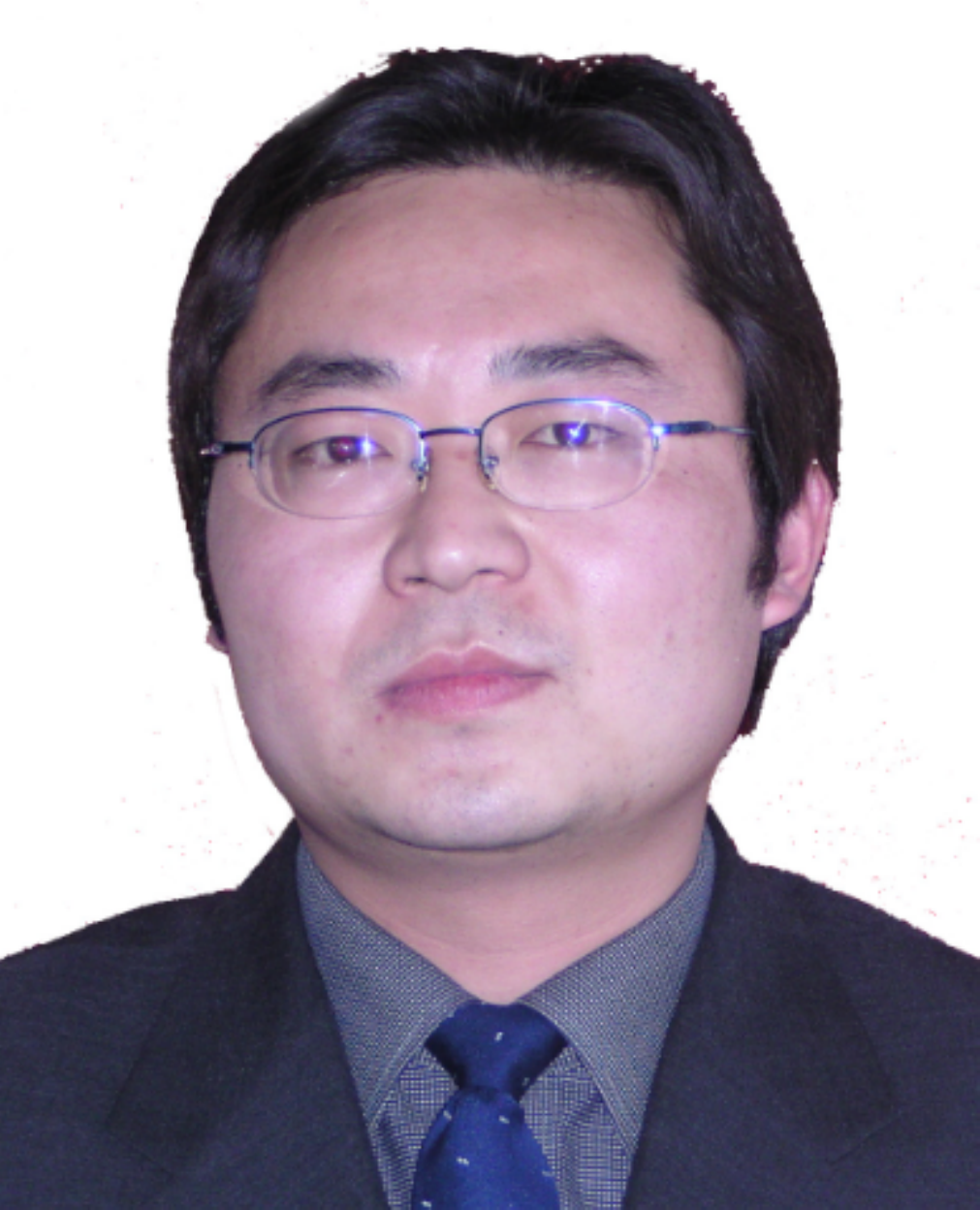}}]{Jiguo Yu} received the Ph.D. degree from the School of Mathematics, Shandong University, in 2004. He became a Full Professor in the School of Computer Science, Qufu Normal University, Shandong, China, in 2007, and currently is a Full Professor at the Qilu University of Technology, Jinan, Shandong China. His main research interests include blockchain, IoT security, privacy-aware computing, wireless distributed computing, and graph theory. He is a Fellow of IEEE, a member of ACM, and a Senior Member of China Computer Federation (CCF).
\end{IEEEbiography}

\end{document}




%

\end{document}